\documentclass[11pt]{article} 

\usepackage{amsmath,amsfonts,amssymb,amsthm,graphicx,color,mathrsfs,dcolumn,natbib,multirow,url}

\usepackage{array, bbm, centernot}

\usepackage{bbm} 

\usepackage[T1]{fontenc} 

\usepackage{algorithm}   
\usepackage[noend]{algpseudocode}

\usepackage[normalem]{ulem} 

\usepackage{hyperref}

\usepackage{subcaption}

\usepackage{setspace}

\usepackage{xr}

\makeatletter
\def\BState{\State\hskip-\ALG@thistlm}
\makeatother

\newcommand{\PreserveBackslash}[1]{\let\temp=\\#1\let\\=\temp}
\newcolumntype{C}[1]{>{\PreserveBackslash\centering}p{#1}}
\newcolumntype{R}[1]{>{\PreserveBackslash\raggedleft}p{#1}}
\newcolumntype{L}[1]{>{\PreserveBackslash\raggedright}p{#1}}

 \newtheorem{theorem}{Theorem}[section]
 
 \newtheorem{corollary}[theorem]{Corollary}
 \theoremstyle{definition}
 \newtheorem{remark}{Remark}[section]

 \numberwithin{equation}{section}

	\usepackage{chngcntr}
\counterwithout{equation}{section} 
\counterwithout{theorem}{section}
\counterwithout{definition}{section}

\newcommand{\N}{\mathbb{N}}

\newcommand{\R}{\mathbb{R}}

\renewcommand{\emptyset}{\text{\usefont{OMS}{cmsy}{m}{n}\symbol{59}}}

\newcommand{\comment}[1]{}

\DeclareMathOperator*{\argmax}{arg\,max}

\renewcommand{\theta}{\vartheta}

\newcommand{\bs}{\boldsymbol}
\newcommand{\out}[1]{{}}

\begin{document}

\allowdisplaybreaks

\baselineskip3.75ex
\title{\bf A Metropolized adaptive subspace algorithm for high-dimensional Bayesian variable selection}

\author{Christian Staerk$^1$, Maria Kateri$^2$ and Ioannis Ntzoufras$^3$ 
        \footnote{
       e-mails: {\tt staerk@imbie.uni-bonn.de}, {\tt maria.kateri@rwth-aachen.de},
 {\tt ntzoufras@aueb.gr}} 
\vspace{2ex}\\ 
$^1$ Department of Medical Biometry, Informatics and Epidemiology, \\ University Hospital Bonn, Germany\\
$^2$ Institute of Statistics, RWTH Aachen University, Germany\\
$^3$ Department of Statistics, Athens University of Economics and Business, Greece
}

\date{}
\maketitle

%

%
\begin{abstract}
 
A simple and efficient adaptive Markov Chain Monte Carlo (MCMC) method, called the Metropolized Adaptive Subspace (MAdaSub) algorithm, is proposed for sampling from high-dimensional posterior model distributions in Bayesian variable selection. The MAdaSub algorithm is based on an independent Metropolis-Hastings sampler, where the individual proposal probabilities of the explanatory variables are updated after each iteration using a form of Bayesian adaptive learning, in a way that they finally converge to the respective covariates' posterior inclusion probabilities. 
We prove the ergodicity of the algorithm and present a parallel version of MAdaSub with an adaptation scheme for the proposal probabilities based on the combination of information from  multiple chains. The effectiveness of the algorithm is demonstrated via various simulated and real data examples, including a high-dimensional problem with more than 20,000 covariates.   
\\[2mm]
\end{abstract}
%
%
\vspace{2mm}
\section{Introduction}\label{sect1}

Variable selection in regression models is one of the big challenges in the era of high-dimensional data where the number of explanatory variables might largely exceed the sample size. During the last two decades, many classical variable selection algorithms have been proposed which are often based on finding the solution to an appropriate optimization problem. As the most famous example, the Lasso \citep{tibshirani1996} relies on an \(\ell_1\)-type relaxation of the original \(\ell_0\)-type optimization problem. 
Convex methods like the Lasso are computationally very efficient and are therefore routinely used in high-dimensional statistical applications. However, such classical methods mainly focus on point estimation and do not provide a measure of uncertainty concerning the best model, per se, although recent works aim at addressing these issues as well (see e.g. \citealp{wasserman2009}, \citealp{meinshausen2010} and \citealp{lee2016}). On the other hand, a major advantage of a fully Bayesian approach is that it automatically accounts for model uncertainty. In particular, Bayesian model averaging \citep{raftery1997} and the median probability model \citep{barbieri2004} can be used for predictive inference. Furthermore, posterior inclusion probabilities of the individual covariates can be computed to quantify the Bayesian evidence. 

Important \(\ell_0\)-type criteria like the Bayesian Information Criterion (BIC, \citealp{schwarz1978}) and the Extended Bayesian Information Criterion (EBIC, \citealp{chen2008}) can be derived as asymptotic approximations to a fully Bayesian approach (compare e.g. \citealp{liang2013}). 
It has been argued that \(\ell_0\)-type methods posses favourable statistical properties in comparison to convex \(\ell_1\)-type methods with respect to variable selection and prediction (see e.g. \citealp{raskutti2011} and \citealp{narisetty2014}). Since solving the associated, generally NP-hard, discrete optimization problems by an exhaustive search is computationally prohibitive, there have been recent attempts in providing more efficient methods for resolving such issues, as for example, mixed integer optimization methods \citep{bertsimas2016} and Adaptive Subspace (AdaSub) methods (\citealp{staerk2018}; \citealp{staerk2021}).   

The challenging practical issue of a fully Bayesian approach is similar to that of optimizing \(\ell_0\)-type information criteria: computing (approximate) posterior model probabilities for all possible models is not feasible if the number of explanatory variables \(p\) is very large, since there are in general \(2^p\) possible models which have to be considered. 
Often, Markov Chain Monte Carlo (MCMC) methods based on Metropolis-Hastings steps (e.g. \citealp{madigan1995}), Gibbs samplers (e.g. \citealp{george1993}; \citealp{dellaportas2002}) and ``reversible jump'' updates (e.g. \citealp{green1995}) are used in order to obtain a representative sample from the posterior model distribution. 
However, the effectiveness of MCMC methods depends heavily on a sensible choice of the proposal distributions being used. Therefore, such methods may suffer from bad mixing resulting in a slow exploration of the model space, especially when the number of covariates is large. Moreover, tuning of the proposal distribution is often only feasible after manual ``pilot'' runs of the algorithm. 

Adaptive MCMC methods aim to address these issues by updating the proposal parameters ``on the fly'' during a single run of the algorithm so that the proposal distribution automatically adjusts according to the currently available information. Recently, a number of different adaptive MCMC algorithms have been proposed in the Bayesian variable selection context, see e.g. \citet{nott2005}, \citet{lamnisos2013}, \citet{ji2013}, \citet{griffin2014}, \citet{griffin2018} and \cite{wan2021}.  
In this work we propose an alternative, simple and efficient adaptive independent Metropolis-Hastings algorithm for Bayesian variable selection, called the Metropolized Adaptive Subspace (MAdaSub) algorithm, and compare it to existing adaptive MCMC algorithms. In MAdaSub the individual proposal probabilities of the explanatory variables are sequentially adapted after each iteration. The employed updating scheme is inspired by the AdaSub method introduced in \citet{staerk2021} and can itself be motivated in a Bayesian way, such that the individual proposal probabilities finally converge against the true respective posterior inclusion probabilities. In the limit, the algorithm can be viewed as a simple Metropolis-Hastings sampler using a product of independent Bernoulli proposals which is the closest to the unknown target distribution in terms of Kullback-Leibler divergence (among the distributions in the family of independent Bernoulli form). 

The paper is structured as follows. The considered setting of Bayesian variable selection in generalized linear models (GLMs) is briefly described in Section~\ref{sec:setting}. The MAdaSub algorithm is motivated and introduced in Section~\ref{sec:MAdaSub}. By making use of general results obtained by \citet{roberts2007}, it is shown that the MAdaSub algorithm is ergodic despite its continuing adaptation, i.e.\ that ``in the limit'' it samples from the targeted posterior model distribution (see Theorem~\ref{thm:MAdaSub}). Alternative adaptive approaches are also briefly discussed and conceptually compared to the newly proposed algorithm. In Section~\ref{sec:parallel}, a parallel version of MAdaSub is presented where the proposal probabilities can be adapted using the information from all available chains, without affecting the ergodicity of the algorithm (see Theorem~\ref{thm:parallel}). 
Detailed proofs of the theoretical results of Sections~\ref{sec:MAdaSub} and~\ref{sec:parallel} can be found in the Supplement to this paper. 
The adaptive behaviour of MAdaSub and the choice of its tuning parameters are illustrated via low- and high-dimensional simulated data applications in Section~\ref{sec:sim}, emphasizing that the speed of convergence against the targeted posterior depends on an appropriate choice of these parameters. In Section~\ref{sec:realdata} various real data applications demonstrate that MAdaSub provides an efficient and stable way for sampling from high-dimensional posterior model distributions. The paper concludes with a discussion in Section~\ref{sec:discussion}. An R-implementation of MAdaSub is available at \url{https://github.com/chstaerk/MAdaSub}.

\section{The setting}\label{sec:setting}

In this work we consider variable selection in univariate generalized linear models (GLMs), where the response variable \(Y\) is modelled in terms of \(p\) possible explanatory variables \(X_1,\dots,X_p\). More precisely, for a sample of size \(n\), the components of the response vector \(\bs Y=(Y_1,\dots,Y_n)^T\) are assumed to be independent with each of them having a distribution from a fixed exponential dispersion family with 
\begin{equation} g\big( E(Y_i \,|\, \bs X_{i,*})\big) = \mu + \sum_{j=1}^p \beta_j X_{i,j},~~~i=1,\dots, n \, , \label{eq:GLM}\end{equation}
where \(g\) is a (fixed) link function, \(\mu\in\R\) is the intercept and \(\bs \beta=(\beta_1,\dots,\beta_p)^T\in\R^{p}\) is the vector of regression coefficients. 
Here, \(\bs X=(X_{i,j})\in\R^{n \times p}\) is the design matrix; it's \(i\)-th row \(\bs X_{i,*}\) corresponds to the \(i\)-th observation and it's \(j\)-th column \(\bs X_{*,j}\equiv \bs X_j\) corresponds to the values of the \(j\)-th predictor. 
For a subset \(S\subseteq\{1,\dots,p\}\), the model induced by \(S\) is defined by a GLM of the form~(\ref{eq:GLM}) but with design matrix \(\bs X_S\in\R^{n\times |S|}\) in place of \(\bs X\in\R^{n\times p}\) and corresponding vector of coefficients \(\bs\beta_S\in\R^{|S|}\), where \(\bs X_S\) denotes the submatrix of the original design matrix \(\bs X\) containing only the columns with indices in \(S\). For brevity, we often simply refer to the model \(S\). Without further notice, we assume that we always include an intercept \(\mu\) in the corresponding GLM with design matrix \(\bs X_{S}\). 
We denote the set of labelled explanatory variables by \(\mathcal{P}=\{1,\dots,p\}\) and the full model space by \( \mathcal{M}=\{S;\, S\subseteq\mathcal{P}\} \).  

In a fully Bayesian approach we assign prior probabilities \(\pi(S)\) to each of the considered models \(S\in\mathcal{M}\) as well as priors \(\pi(\mu,\psi,\bs\beta_S \,|\, S)\) for the parameters of each model \(S\in\mathcal{M}\), where \(\psi\) denotes a possibly present dispersion parameter (e.g.\ the variance in a normal linear model). After observing data \(\mathcal{D}=(\bs X,\bs y)\), with \(\bs X\in\R^{n\times p}\) and \(\bs y\in\R^n\), the posterior model probabilities are proportional to
\begin{equation}
\pi( S \,|\, \mathcal{D} ) \propto  \pi( \bs y \,|\,\bs X, S) \, \pi( S )  \, ,   ~ S\in\mathcal{M} \,,
\end{equation}
where 
\begin{equation}
\pi( \bs y \,|\,\bs X, S) = \int\int\int f(\bs y \,|\,\bs X, S, \mu,\psi, \bs\beta_S) \, \pi(\mu,\psi, \bs\beta_S \,|\, S) \, d\mu \, d\psi  \, d\bs\beta_S \, 
\end{equation}
is the marginal likelihood of the data \(\bs y\) under model \(S\), while \(f(\bs y \,|\,\bs X, S, \mu,\psi, \bs\beta_S) \) denotes the likelihood of the data \(\bs y\) under model \(S\) given the parameter values \(\mu,\psi,\bs\beta_S\) and the values of the explanatory variables \(\bs X\). 
Note that the marginal likelihood \(\pi( \bs y \,|\,\bs X, S)  \) is generally only available in closed form when conjugate priors are used. 

\begin{remark}\label{remark:conjugate} A prominent example in normal linear models is a conjugate prior structure, where the prior on the variance \(\psi = \sigma^2\) is given by Jeffreys prior (independent of the model~\(S\)) and the prior on the vector of coefficients \(\bs\beta_S\) in model \(S\in\mathcal{M}\) is given by a multivariate normal distribution, i.e. 
\begin{equation}\label{eq:prior1}
 \bs \beta_S \, | \, S ,  \sigma^2 \sim \mathcal{N}_{|S|} (\bs\theta_S , \sigma^2 g \, {\bs W}_{\!\! S} ) , ~~~ \pi(\sigma^2) \propto \frac{1}{\sigma^2} \,, 
\end{equation}
where \( \bs\theta_S\in\R^{|S|}\), \(g>0\) and \({\bs W}_{\!\! S}\in \R^{|S|\times |S|}\) 
are hyperparameters. After centering each of the covariates \(\bs X_j\), \(j\in\mathcal{P}\), the improper prior \(\pi(\mu)\propto 1\) is a common choice for the intercept~\(\mu\) (again, independent of the model \(S\)). With no specific prior information, the prior mean of~\(\bs\beta_S\) can be set to the zero vector (\(\bs\theta_S= \bs 0\)). The 
matrix \(\bs {\bs W}_{\!\! S}\) is often chosen to be the identity matrix~\(\bs I_{|S|}\) of dimension \(|S|\) 
or to be \( {\bs W}_{\!\! S} = (\bs X_S^T \bs X_S)^{-1} \) yielding Zellner's g-prior \citep{zellner1986}. 
The first choice corresponds to Ridge Regression and implies prior independence of the regression coefficients, while the second choice with \(g=n\) corresponds to a unit information prior. In case no specific prior information is available about the possible regressors, a natural choice for the model prior is an independent Bernoulli prior of the form
\begin{equation}\label{eq:modelprior}
\pi(S\,|\,\omega) = \omega^{|S|} (1-\omega)^{p-|S|} ,\, S\in\mathcal{M} \,, 
\end{equation}
where \(\omega=\pi(j\in S)\) is the prior probability that variable \(X_j\) is included in the model, for all \(j\in\mathcal{P}\). One can either set the prior inclusion probability \(\omega\) to some fixed value or consider an additional hyperprior for \(\omega\), with the latter option yielding more flexibility. A convenient choice is the (conjugate) beta prior 
\(\omega \sim \mathcal{B}e(a_\omega,b_\omega)\), where \(a_\omega>0\) and \(b_\omega>0\) can be chosen in order to reflect the prior expectation and prior variance of the model size \(s=|S|\), \(S\in\mathcal{M}\) (see \citealp{kohn2001} for details). In practice, one often imposes an a-priori upper bound~$s_{\text{max}}$ on the model size (with $s_{\text{max}}\leq n$) by setting \(\pi(S)=0\) for \(|S|> s_{\text{max}}\) (cf.~\citealp{liang2013, rossell2021}), while for fixed control variables~$X_j$ one can enforce the inclusion of such variables by setting \(\pi(j\in S)=1\). 
\end{remark}

In the general non-conjugate case the marginal likelihood is not readily computable and numerical methods may be used for deriving an approximation to the marginal likelihood. Laplace's method yields an asymptotic analytic approximation to the marginal likelihood \citep{kass1995}. Similarly, different information criteria like the Bayesian Information Criterion (BIC, \citealp{schwarz1978}) or the Extended Bayesian Information Criterion (EBIC, \citealp{chen2008}) can be used directly as asymptotic approximations to fully Bayesian posterior model probabilities under suitable choices of model priors. Under a uniform model prior, i.e. \(\pi(S)=\frac{1}{2^p}\) for all \(S\in\mathcal{M}\), the BIC can be derived as an approximation to \(-2\log(\text{BF}(S)) = -2\log(\text{PO}(S))\), where \(\text{BF}(S) = \pi( \bs y \,|\,\bs X , S)/ \pi( \bs y \,|\,\bs X , \emptyset) \) denotes the Bayes factor of model \(S\in\mathcal{M}\) versus the null model \(\emptyset\in\mathcal{M}\) and \(\text{PO}(S)\) 
denotes the corresponding posterior odds  (\citealp{schwarz1978}; \citealp{kass1995b}). 
In a high-dimensional but sparse situation, in which only a few of the many possible predictors contribute substantially to the response, a uniform prior on the model space is a naive choice since it induces severe overfitting. Therefore, \citet{chen2008} propose the prior  
\begin{equation}\pi(S)\propto {\binom{p}{|S|}}^{-\gamma}\,,\label{eq:EBICprior}\end{equation}
where \(\gamma\in[0,1]\) is an additional parameter. 
If~\(\gamma=1\), then \(\pi(S)=\frac{1}{p+1}{ \binom{p}{|S|}}^{-1}\), so the prior gives equal probability to each model size, and to each model of the same size; note that this prior does also coincide with the beta-binomial model prior discussed above when setting \(a_\omega=b_\omega=1\), providing automatic multiplicity correction \citep{scott2010}. If~\(\gamma=0\), then we obtain the uniform prior used in the original BIC. Similar to the derivation of the BIC one asymptotically obtains the EBIC with parameter \(\gamma\in[0,1]\) as
\begin{equation}\text{EBIC}_\gamma(S) = -2\log\left( f(\bs y \,|\,\bs X, S, \hat{\mu}_S,\hat{\psi}_S, \hat{\bs\beta}_S) \right)+ \Big(\log(n)+2\gamma\log(p)\Big) |S| \,,\label{def:EBIC2}\end{equation}
where \(f(\bs y \,|\,\bs X, S, \hat{\mu}_S,\hat{\psi}_S, \hat{\bs\beta}_S) \) denotes the maximized likelihood under the model \(S\in\mathcal{M}\) (compare \citealp{chen2012}). Under the model prior~(\ref{eq:EBICprior}) and a unit-information prior on the regression coefficients for each model \(S\in\mathcal{M}\), one can asymptotically approximate the model posterior by 
 \begin{equation} \pi(S\,|\,\mathcal{D}) \approx \frac{\exp\left(-\frac{1}{2}\times \text{EBIC}_\gamma(S)\right)}{\sum_{S'\in\mathcal{M}} \exp\left(-\frac{1}{2}\times \text{EBIC}_\gamma(S')\right)} \,, ~ S\in\mathcal{M} \, . \label{eq:EBICkernel}\end{equation} 
In this work we consider situations where the marginal likelihood \(\pi( \bs y \,|\,\bs X , S)\) is available in closed form due to the use of conjugate priors (see Remark~\ref{remark:conjugate}) or where an approximation to the posterior~\(\pi(S\,|\,\mathcal{D})\) is used (e.g. via  equation~(\ref{eq:EBICkernel}) with the EBIC or any other \(\ell_0\)-type criteria such as the risk inflation criterion, cf.~\citealp{foster1994, rossell2021}).  
This assumption allows one to focus on the essential part of efficient sampling in very large model spaces, avoiding challenging technicalities regarding sampling of model parameters for non-conjugate cases. It also facilitates empirical comparisons with other recent adaptive variable selection methods, which focus on conjugate priors~\citep{zanella2019, griffin2018}. 
Furthermore, conjugate priors such as the g-prior as well as normalized \(\ell_0\)-type selection criteria such as the EBIC in equation~(\ref{eq:EBICkernel}) have shown to provide concentration of posterior model probabilities on the (Kullback-Leibler) optimal model under general conditions even in case of model misspecification~\citep{rossell2021}, as well as model selection consistency for the true model in GLMs without misspecification~\citep{chen2012, liang2013}.

\section{The MAdaSub algorithm}\label{sec:MAdaSub}

A simple way to sample from a given target distribution is to use an independent Metropolis-Hastings algorithm. Clearly, the efficiency of such an MCMC algorithm depends on the choice of the proposal distribution, which is in general not an easy task (see e.g.\ \citealp{rosenthal2011}).  
In the ideal situation, the proposal distribution for an independence sampler should be the same as the target distribution  \(\pi(S\,|\,\mathcal{D})\), leading to an independent sample from the target distribution with corresponding acceptance probability of one. Adaptive MCMC algorithms aim to sequentially update the proposal distribution during the algorithm based on the previous samples such that, in case of the independence sampler, the proposal becomes closer and closer to the target distribution as the MCMC sample grows (see e.g.\ \citealp{holden2009}, \citealp{giordani2010}). However, especially in high-dimensional situations, it is crucial that the adaptation of the proposal as well as sampling from the proposal can be carried out efficiently. For this reason, we restrict ourselves to proposal distributions which have an independent Bernoulli form:\ if \(S\in\mathcal{M}\) is the current model, then we propose model \(V\in\mathcal{M}\) with probability
\begin{equation}\label{eq:inde}
q(V\,|\,S;\bs r) \equiv q(V;\bs r) = \prod_{j\in V} r_j \prod_{j\in\mathcal{P}\setminus V} (1-r_j) \,, 
\end{equation} 
for some vector \(\bs r=(r_1,\dots,r_p)\in (0,1)^p\) of individual proposal probabilities.

\begin{algorithm}[]
\caption{Metropolized Adaptive Subspace (MAdaSub) algorithm}\label{algo:MCMC} 
\begin{flushleft} \textbf{Input:} \end{flushleft}
\vspace{-7mm}
\begin{itemize}
\item Data $\mathcal{D}=(\bs X,\bs y)$.
\vspace{-2mm}
\item (Approximate) kernel of posterior \(\pi(S\,|\,\mathcal{D})\propto \pi( \bs y \,|\,\bs X, S) \, \pi( S ) \) for \(S\in\mathcal{M}\).
\vspace{-2mm}
\item Vector of initial proposal probabilities \(\bs r^{(0)} = \left(r_1^{(0)},\dots,r_p^{(0)}\right)^T\in(0,1)^p\). 
\vspace{-2mm}
\item Parameters $L_j>0$  for \(j\in\mathcal{P}\), controlling the adaptation rate of the algorithm (e.g.\ $L_j=L=p$).
\vspace{-2mm}
\item Constant \(\epsilon\in(0,0.5)\) (chosen to be small, e.g.\ \(\epsilon\leq\frac{1}{p}\)).
\vspace{-2mm}
\item Number of iterations $T\in\N$.
\vspace{-2mm}
\item Starting point \(S^{(0)}\in\mathcal{M}\) (optional).
\end{itemize}

\vspace{-5mm}
\begin{flushleft} \textbf{Algorithm: } \end{flushleft}
\vspace{-5mm}
\begin{enumerate}
\vspace{-2mm}
\item[(1)] 
If starting point \(S^{(0)}\) not specified:
\vspace{-1mm}  
\par
\begingroup
\leftskip=0.5cm 
\noindent Sample $b_j^{(0)}\sim\text{Bernoulli}\left(r_j^{(0)}\right)$ independently for $j\in\mathcal{P}$. \\
Set \(S^{(0)} = \{j\in\mathcal{P};~b_j^{(0)}=1\}\).
\par
\endgroup 
\item[(2)] For $t=1,\dots,T$:
\vspace{-2mm}
\begin{enumerate}
\item[(a)] Truncate vector of proposal probabilities to 
\(\tilde{\bs r}^{(t-1)} = \left(\tilde{r}_1^{(t-1)},\dots,\tilde{r}_p^{(t-1)}\right)^T\), i.e. for \(j\in\mathcal{P}\) set 
\vspace{-1mm}
\begin{align*} \tilde{r}_j^{(t-1)} = \begin{cases} r_j^{(t-1)} &\mbox{, if } r_j^{(t-1)}\in[\epsilon,1-\epsilon] \,, \\
\epsilon & \mbox{, if } r_j^{(t-1)}<\epsilon \,,\\ 
1-\epsilon & \mbox{, if } r_j^{(t-1)}>1-\epsilon \,. \\ \end{cases} \end{align*}
\vspace{-3mm}
\item[(b)] Draw $b_j^{(t)}\sim\text{Bernoulli}\left(\tilde{r}_j^{(t-1)}\right)$ independently for $j\in\mathcal{P}$.
\vspace{-1mm}
\item[(c)] Set $V^{(t)}=\{j\in\mathcal{P};~b_j^{(t)}=1\}$.
\vspace{-1mm}
\item[(d)] Compute acceptance probability \[ \alpha^{(t)} 
= \min\left\{ \frac{\pi( \bs y \,|\,\bs X, V^{(t)}) \, \pi( V^{(t)} ) \, q(S^{(t-1)}; \tilde{\bs r}^{(t-1)})}  { \pi( \bs y \,|\,\bs X, S^{(t-1)} ) \, \pi( S^{(t-1)} ) \, q(V^{(t)}; \tilde{\bs r}^{(t-1)} )} ,\, 1\right\} \,.\]
\vspace{-3mm}
\item[(e)] Set \(S^{(t)} = \begin{cases} V^{(t)} &\mbox{, with probability } \alpha^{(t)}, \\
S^{(t-1)} & \mbox{, with probability } 1-\alpha^{(t)}. \\ \end{cases} \)
\vspace{-1mm}
\item[(f)] Update vector of proposal probabilities \(\bs r^{(t)} = \left(r_1^{(t)},\dots,r_p^{(t)}\right)^T\) via 
\vspace{-1mm}
\[r_j^{(t)}= \frac{L_j r_j^{(0)}+\sum_{i=1}^t \mathbbm{1}_{S^{(i)}}(j)}{L_j+t} \,,~~ j\in\mathcal{P} \, . \]
\end{enumerate}
\end{enumerate}

\vspace{-5mm}
\begin{flushleft} \textbf{Output:}  \end{flushleft}
\vspace{-7mm}

\begin{itemize}
\item Approximate sample \(S^{(b+1)},\dots,S^{(T)}\) from posterior distribution \(\pi(\cdot\,|\,\mathcal{D})\), after burn-in period of length \(b\).
\end{itemize}
\end{algorithm}

\subsection{Serial version of the MAdaSub algorithm}\label{sec:3.1}

The fundamental idea of the newly proposed MAdaSub algorithm (given below as Algorithm~\ref{algo:MCMC}) is to sequentially update the individual proposal probabilities according to the currently ``estimated'' posterior inclusion probabilities. In more detail, after initializing the vector of proposal probabilities \(\bs r^{(0)} = \left(r_1^{(0)},\dots,r_p^{(0)}\right)\in(0,1)^p\), the individual proposal probabilities \(r_j^{(t)}\) of variables \(X_j\) are updated after each iteration \(t\) of the algorithm, such that \(r_j^{(t)}\) finally converges to the actual posterior inclusion probability \(\pi_j = \pi(j\in S\,|\,\mathcal{D})\), as \(t\rightarrow\infty\) (see Corollary~\ref{cor:MAdaSub} below). Therefore, in the limit, we make use of the proposal 
\begin{equation} \label{eq:limitingprop}
q(V;\bs r^*) = \prod_{j\in V} \pi_j \prod_{j\in\mathcal{P}\setminus V} (1-\pi_j) ,~~ V\in\mathcal{M} \,, ~~ \text{with } \bs r^* = (\pi_1,\dots,\pi_p) \, ,
\end{equation}
which is the closest distribution (in terms of Kullback-Leibler divergence) to the actual target \(\pi(S\,|\,\mathcal{D})\), among all distributions of independent Bernoulli form~(\ref{eq:inde}) (see \citealp{clyde2011}). Note that the median probability model \citep{barbieri2004, barbieri2021}, defined by \(S_{\text{MPM}}=\{j\in\mathcal{P}:\pi_j\geq 0.5\}\), has the largest probability in the limiting proposal~(\ref{eq:limitingprop}) of MAdaSub, i.e.\ \(\argmax_{V\in\mathcal{M}} q(V;\bs r^*) = S_{\text{MPM}}\). Thus, MAdaSub can be interpreted as an adaptive algorithm which aims to adjust the proposal so that models in the region of the median probability model are proposed with increasing probability. 

For \(j\in\mathcal{P}\), the concrete update of \(r_j^{(t)}\) after iteration \(t\in\N\) is given by
\begin{equation} \label{eq:update}
r_j^{(t)} \,=\, \frac{L_j r_j^{(0)}+\sum_{i=1}^t \mathbbm{1}_{S^{(i)}}(j)}{L_j+t} \,=\, \left(1-\frac{1}{L_j+t}\right)r_j^{(t-1)} + \frac{\mathbbm{1}_{S^{(t)}}(j)}{L_j+t} \,,
\end{equation}
where, for \(j\in\mathcal{P}\), \(L_j>0\) are additional parameters controlling the adaptation rate of the algorithm and \(\mathbbm{1}_{S^{(i)}}\) denotes the indicator function of the set \(S^{(i)}\). If \(j\in S^{(t)}\) (i.e.\ \(\mathbbm{1}_{S^{(t)}}(j)=1\)), then variable~\(X_j\) is included in the sampled model in iteration \(t\) of the algorithm and the proposal probability \(r_j^{(t)}\) of \(X_j\) increases in the next iteration~\(t+1\); similarly, if \(j\notin S^{(t)}\) (i.e.\ \(\mathbbm{1}_{S^{(t)}}(j)=0\)), then the proposal probability decreases. The additional ``truncation'' step 2 (a) in the MAdaSub algorithm ensures that the truncated individual proposal probabilities \(\tilde{r}_j^{(t)}\), \(j\in\mathcal{P}\), are always included in the compact interval \(\mathcal{I}=[\epsilon,1-\epsilon]\), where \(\epsilon\in(0,0.5)\) is a pre-specified ``precision'' parameter. 
This adjustment simplifies the proof of the ergodicity of MAdaSub. 
Note that the mean size of the proposed model~\(V\) from the proposal \(q(V;\tilde{\bs r})\) in equation~(\ref{eq:inde}) with \(\tilde{\bs r}\in[\epsilon,1-\epsilon]^p\) is at least \(E|V|\geq\epsilon \times p\); thus, in practice we recommended to set \(\epsilon\leq\frac{1}{p}\), so that models of small size including the null model can be proposed with sufficiently large probability. On the other hand, if \(\epsilon\) is chosen to be very small, then the MAdaSub algorithm may take a longer time to convergence in case proposal probabilities of informative variables are close to \(\epsilon\approx 0\) during the initial burn-in period of the algorithm. Simulations and real data applications show that the choice \(\epsilon=\frac{1}{p}\) works well in all considered situations (see Sections~\ref{sec:sim} and~\ref{sec:realdata}).   

The updating scheme of the individual proposal probabilities is inspired by the AdaSub method proposed in \citet{staerk2018} and \citet{staerk2021} and can itself be motivated in a Bayesian way: since we do not know the true posterior inclusion probability \(\pi_j\) of variable \(X_j\) for \(j\in\mathcal{P}\), we place a beta prior on \(\pi_j\) with the following parametrization
\begin{equation}
\pi_j \sim \mathcal{B}e\left(L_j r_j^{(0)},\, L_j \left(1- r_j^{(0)}\right)\right) \, ,
\label{eq:betaprior}
\end{equation} 
where \(r_j^{(0)}=E[\pi_j]\) is the prior expectation of \(\pi_j\) and \(L_j>0\) controls the variance of \(\pi_j\) via
\begin{equation}
\text{Var}(\pi_j) = \frac{1}{L_j+1} \times r_j^{(0)}\, \left(1-r_j^{(0)}\right) \,.
\end{equation}
If \(L_j\rightarrow 0\), then \(\text{Var}(\pi_j) \rightarrow r_j^{(0)} \, (1-r_j^{(0)}) \), which is the variance of a Bernoulli random variable with mean \(r_j^{(0)}\). If \(L_j\rightarrow \infty\), then \(\text{Var}(\pi_j) \rightarrow 0\). 
Now, one might view the samples \(S^{(1)},\dots,S^{(t)}\) obtained after \(t\) iterations of MAdaSub as ``new'' data and interpret the information learned about \(\pi_j\) as \(t\) approximately independent Bernoulli trials, where \(j\in S^{(i)}\) corresponds to ``success'' and \(j\notin S^{(i)}\) corresponds to ``failure''. Then the (pseudo) posterior of \(\pi_j\) after iteration \(t\) of the algorithm is given by 
\begin{equation}
\pi_j \,|\,S^{(1)},\dots,S^{(t)}  \sim \mathcal{B}e\left( L_j r_j^{(0)} + \sum_{i=1}^t \mathbbm{1}_{S^{(i)}}(j),\, L_j (1- r_j^{(0)}) + \sum_{i=1}^t \mathbbm{1}_{\mathcal{P}\setminus S^{(i)}}(j) \right) \,  ,
\label{eq:posterior}
\end{equation} 
with posterior expectation 
\begin{equation}\label{eq:update_again}
E(\pi_j \,|\, S^{(1)},\dots,S^{(t)}) = \frac{L_j r_j^{(0)}+\sum_{i=1}^t \mathbbm{1}_{S^{(i)}}(j)}{L_j+t} = r_j^{(t)} \,
\end{equation}
and posterior variance 
\begin{equation}\label{eq:updated_var}
\text{Var}(\pi_j\,|\, S^{(1)},\dots,S^{(t)}) = \frac{1}{L_j+t+1} \times r_j^{(t)}\, \left(1-r_j^{(t)}\right) \,.
\end{equation}

The interpretation of~\(r_j^{(0)}\) as the prior expectation for the posterior inclusion probability~\(\pi_j\) motivates the choice of \(r_j^{(0)}= \pi(j\in S)\) as the actual prior inclusion probability of variable~\(X_j\). 
If no particular prior information about specific variables is available, but the prior expected model size is equal to $q\in(0,p)$, then we recommend to set \(r_j^{(0)}=\frac{q}{p}\) and  \(L=L_j=p\) for all \(j\in\mathcal{P}\), corresponding to the prior~\(\pi_j\sim\mathcal{B}e(q,p-q)\) in equation~\eqref{eq:betaprior}. In this particular situation, equation~(\ref{eq:update_again}) reduces to
\begin{equation}\label{eq:update_no}
E(\pi_j \,|\, S^{(1)},\dots,S^{(t)}) = \frac{q+\sum_{i=1}^t \mathbbm{1}_{S^{(i)}}(j)}{p+t} = r_j^{(t)} \,.
\end{equation}

Even though it seems natural to choose the parameters \(r_j^{(0)}\) and \(L_j\) of MAdaSub as the respective prior quantities, this choice is not imperative. While the optimal choices of these parameters generally depend on the setting, various simulated and real data applications of MAdaSub indicate that choosing \(r_j^{(0)}=\frac{q}{p}\) with \(q\in[2,10]\) and \(L_j\in[p/2,2p]\) for \(j\in\mathcal{P}\) yields a stable algorithm with good mixing in sparse high-dimensional set-ups irrespective of the actual prior (see Sections~\ref{sec:sim} and~\ref{sec:realdata}). Furthermore, if one has already run and stopped the MAdaSub algorithm after a certain number of iterations~\(T\), then one can simply restart the algorithm with the already updated parameters~\(r_j^{(T)}\) and~\(L_j+T\) (compare equation~(\ref{eq:updated_var})) as new starting values for the corresponding parameters. 

Using general results for adaptive MCMC algorithms by \citet{roberts2007}, we show that MAdaSub is ergodic despite its continuing adaptation. 

\begin{theorem}\label{thm:MAdaSub}
The MAdaSub algorithm (Algorithm~\ref{algo:MCMC}) is ergodic for all choices of \(\bs r^{(0)}\in(0,1)^p\), \(L_j>0\) and \(\epsilon\in(0,0.5)\) and fulfils the weak law of large numbers. 
\end{theorem}

The proof of Theorem~\ref{thm:MAdaSub} can be found in Section~A of the Supplement, where it is shown that MAdaSub satisfies both the simultaneous uniform ergodicity condition and the diminishing adaptation condition (cf.\ \citealp{roberts2007}). As an immediate consequence of Theorem~\ref{thm:MAdaSub} we obtain the following important result. 

\begin{corollary}\label{cor:MAdaSub}
For all choices of \(\bs r^{(0)}\in(0,1)^p\), \(L_j>0\) and \(\epsilon\in(0,0.5)\), the proposal probabilities \(r_j^{(t)}\) of the explanatory variables \(X_j\) in MAdaSub converge (in probability) to the respective posterior inclusion probabilities \(\pi_j=\pi(j\in S\,|\,\mathcal{D})\), i.e. for all \(j\in\mathcal{P}\) it holds that 
\( r_j^{(t)} \overset{\text{P}}{\rightarrow} \pi_j \) as \(t\rightarrow\infty\).
\end{corollary}

\subsection{Comparison to related adaptive approaches}\label{sec:comp}

In this section we conceptually compare the proposed MAdaSub algorithm (Algorithm~\ref{algo:MCMC}) with other approaches for high-dimensional Bayesian variable selection, focusing on adaptive MCMC algorithms most closely related to the new algorithm (see Section~D of the Supplement for details on further related methods). 

In a pioneering work, \citet{nott2005} propose an adaptive sampling algorithm for Bayesian variable selection based on a  Metropolized Gibbs sampler, showing empirically that the adaptive algorithm outperforms different non-adaptive algorithms in terms of efficiency per iteration. However, since their approach requires the computation of inverses of estimated covariance matrices, it does not scale well to very high-dimensional settings.   
Recently, several variants and extensions of the original adaptive MCMC sampler of \citet{nott2005} have been developed, including an adaptive Metropolis-Hastings algorithm by \citet{lamnisos2013}, where the expected number of variables to be changed by the proposal is adapted during the algorithm. \citet{zanella2019} propose a tempered Gibbs sampling algorithm with adaptive choices of components to be updated in each iteration. Furthermore, different individual adaptation algorithms have been developed in \citet{griffin2014} as well as in the follow-up works of \citet{griffin2018} and \citet{wan2021}, which are closely related to the proposed MAdaSub algorithm. 
These strategies are based on adaptive Metropolis-Hastings algorithms, where the employed proposal distributions are of the following form: if \(S\in\mathcal{M}\) is the current model, then the probability of proposing the model \(V\in\mathcal{M}\) is given by
\begin{equation} \label{eq:propalt}
\tilde{q}(V\,|\,S;\bs \eta) =  \prod_{j\in V\setminus S} A_j  \prod_{j\in S\setminus V} D_j 
\prod_{j\in \mathcal{P}\setminus( S\cup V)} (1-A_j)  \prod_{j\in  S\cap V} (1-D_j) \, ,
\end{equation}
where \(\bs\eta=(A_1,\dots,A_p,D_1,\dots,D_p)^T\in(0,1)^{2p}\) is a vector of tuning parameters with the following interpretation: For \(j\in\mathcal{P}\), \(A_j\) is the probability of adding variable~\(X_j\) if it is not included in the current model \(S\) and \(D_j\) is the probability of deleting variable \(X_j\) if it is included in the current model \(S\). An important difference is that the adaptation strategies in \citet{griffin2018} specifically aim to guard against low acceptance rates of the proposal~(\ref{eq:propalt}), while MAdaSub aims at obtaining a global independent proposal with the largest possible acceptance rate, focusing on regions close to the median probability model. Furthermore, the adaptation of the individual proposal probabilities in MAdaSub can be motivated in a Bayesian way, leading to a natural parallel implementation of the algorithm with an efficient joint updating scheme for the shared adaptive parameters (see Section~\ref{sec:parallel}). Finally, in contrast to MAdaSub, \cite{griffin2018} make use of Rao-Blackwellized estimates of posterior inclusion probabilities to speed up convergence. 

\citet{schafer2013} develop sequential Monte Carlo algorithms (cf.~\citealp{south2019}) 
using model proposals which directly account for the non-independent posterior inclusion of covariates. In contrast, MAdaSub is an adaptive MCMC algorithm which is based on independent Bernoulli proposals. While similar extensions of MAdaSub might be desirable to better approximate the posterior distribution, this may come at the price of a larger computational cost for updating and sampling from the proposal. The simple independent Bernoulli proposals in MAdaSub can also be viewed as mean-field variational approximations to the full posterior model distribution. Despite its connection with variational Bayes approaches (e.g.~\citealp{carbonetto2012}; \citealp{ormerod2017}), MAdaSub samples from the full posterior distribution and the accuracy of the approximation only affects the efficiency of the sampler, as final acceptance rates are expected to be smaller for larger distances between the posterior and the closest independent Bernoulli proposal (cf.~\citealp{neklyudov2019}). Empirical results for MAdaSub (see Sections~\ref{sec:sim} and~\ref{sec:realdata}) indicate that even the simple independent Bernoulli proposals yield good mixing and sufficiently large acceptance rates in various settings.

Finally, MAdaSub is an extension of the Adaptive Subspace (AdaSub) method \citep{staerk2021}, a stochastic search algorithm aiming to identify the best model according to a particular selection criterion (such as the EBIC) by adaptively solving low-dimensional sub-problems of the original problem. While the purpose of AdaSub is to obtain the solution to an optimization problem, its Metropolized version MAdaSub constitutes an adaptive MCMC algorithm which samples from the full posterior model distribution. Despite this difference, the adaptation scheme of AdaSub for the covariates' inclusion probabilities in the sub-problems can be similarly motivated in a Bayesian way (cf.\ \citealp{staerk2018}). The adaptation in AdaSub and MAdaSub is also related to Thompson sampling for multi-armed bandits in reinforcement learning, which has recently been investigated in the context of non-parametric Bayesian variable selection~\citep{liu2021}. In contrast to MAdaSub, Thompson Variable Selection~(TVS) does not provide samples from the posterior distribution but is designed to minimize the regret (i.e.\ the difference between optimal and actual rewards); as a consequence, the sampling probabilities in TVS are not guaranteed to converge to the posterior inclusion probabilities.

\section{Parallelization of the MAdaSub algorithm}\label{sec:parallel}

In this section we present a parallel version of the MAdaSub algorithm which aims at increasing the computational efficiency and accelerating the convergence of the chains. The simplest approach to parallelization would be to independently run the MAdaSub algorithm in parallel on each of \(K\in\N\) different workers, yielding \(K\) individual chains which, in the limit, sample from the posterior model distribution (see Theorem~\ref{thm:MAdaSub}). However, it is desirable that the information learned about the adaptive parameters can be shared efficiently between the different chains, so that the convergence of the adaptive parameters to their optimal values can be accelerated, leading to a faster convergence of the chains to their common limiting distribution. 

We propose a parallel version of MAdaSub, where the workers sample individual MAdaSub chains in parallel, but the acquired information is exchanged periodically between the chains and the adaptive proposal probabilities are updated together (see Algorithm~2 in Section~B of the Supplement for full algorithmic details). More specifically, let \(S^{(k,1)},\dots,S^{(k,T)}\) denote the models sampled by MAdaSub (see Algorithm~\ref{algo:MCMC}) for the first \(T\)~iterations on worker \(k\), for \(k\in\{1,\dots,K\}\). Then, for each worker \(k\in\{1,\dots,K\}\), we define the jointly updated proposal probabilities after the first round~(\(m=1\)) of \(T\)~iterations by 
\begin{equation} \bar{ r}_j^{(k,1)} = \frac{ L_j^{(k)} r_j^{(k,0)} + \sum_{t=1}^{T} \sum_{l=1}^K \mathbbm{1}_{S^{(l,t)}}(j) }{ L_j^{(k)} + TK } ,~~j\in\mathcal{P} \,,\end{equation}
where \(r_j^{(k,0)}\) denotes the initial proposal probability for variable \(X_j\) and \(L_j^{(k)}\) the corresponding adaptation parameter (both can be different across the chains). 

After the joint update, each MAdaSub chain is resumed (with \(\bar{ r}_j^{(k,1)}\) as initial proposal probabilities and \(L^{(k)}_j + TK\) as initial prior variance parameters for \(j\in\mathcal{P}\)) and is run independently on each of the workers for \(T\) additional iterations in a second round (\(m=2\)); then the proposal probabilities are updated jointly again to \(\bar{ r}_j^{(k,2)}\), and so on (up to \(m=R\) rounds in Algorithm~2 of the Supplement). The joint updates of the proposal probabilities after \(m\in\N\) rounds of \(T\) iterations are given by 
\begin{equation} \bar{ r}_j^{(k,m)} = \frac{ L_j^{(k)} r_j^{(k,0)} + \sum_{t=1}^{mT} \sum_{l=1}^{K} \mathbbm{1}_{S^{(l,t)}}(j) }{ L_j^{(k)} + mTK } ,~~k\in\{1,\dots,K\}, ~ j\in\mathcal{P} . \label{eq:jointupdate}\end{equation}

Similarly to the serial version of MAdaSub, the adaptive learning of its parallel version can be naturally motivated in a Bayesian way: 
each worker \(k=1,\dots,K\) can be thought of as an individual subject continuously updating its prior belief about the true posterior inclusion probability \(\pi_j\) of variable \(X_j\) through new information from its individual chain; additionally, after a period of \(T\) iterations the subject updates its prior belief also by obtaining new information from the \(K-1\) other subjects. If the (possibly different) priors of subjects \(k=1,\dots,K\) on \(\pi_j\) are 
\begin{equation}
\pi_j \sim \mathcal{B}e\left(L_j^{(k)} r_j^{(k,0)},\, L_j^{(k)} \left(1- r_j^{(k,0)}\right)\right), \, j\in\mathcal{P}\,,
\end{equation}
where \(r_j^{(k,0)}=E[\pi_j]\) is the prior expectation of subject \(k\) about \(\pi_j\) and \(L_j^{(k)}>0\) controls its prior variance, then the (pseudo) posterior of subject \(k\) about \(\pi_j\) after \(m\) rounds of \(T\) iterations of the parallel MAdaSub algorithm 
is given by (compare to equation~(\ref{eq:posterior})) 
\begin{align}
\pi_j \,\big|\,S^{(1,1)},\dots,S^{(k,mT)}  \sim \mathcal{B}e\bigg( &  L_j^{(k)} r_j^{(k,0)} + \sum_{i=1}^{mT} \sum_{l=1}^K \mathbbm{1}_{S^{(l,i)}}(j) ,\, \nonumber  \\ 
 & L_j^{(k)} (1- r_j^{(k,0)}) + \sum_{i=1}^{mT} \sum_{l=1}^K  \mathbbm{1}_{\mathcal{P}\setminus S^{(l,i)}}(j) \bigg) 
\end{align}
with posterior expectation (compare to equation~(\ref{eq:update_again}))
\begin{equation}
E(\pi_j \,|\, S^{(1,1)},\dots,S^{(k,mT)}) =  \bar{r}_j^{(k,m)} \,,
\end{equation}
corresponding to the joint update in equation~(\ref{eq:jointupdate}).

Although the individual chains in the parallel MAdaSub algorithm make use of the information from all the other chains in order to update the proposal parameters, the ergodicity of the chains is not affected. 

\begin{theorem}\label{thm:parallel}
Consider the parallel version of MAdaSub (see Algorithm~2 in the Supplement). Then, for each worker \(k\in\{1,\dots,K\}\) and all choices of \(\bs r^{(k,0)}\in(0,1)^p\), \(L^{(k)}_j>0\), \(j\in\mathcal{P}\) and \(\epsilon\in(0,0.5)\), each induced chain \(S^{(k,0)},S^{(k,1)},\dots\) of the workers \(k=1,\dots,K\) is ergodic and fulfils the weak law of large numbers. 
\end{theorem}

\begin{corollary}\label{cor:Mparallel} 
For each worker \(k\in\{1,\dots,K\}\) and all choices of \(\bs r^{(k,0)}\in(0,1)^p\), \(L^{(k)}_j>0\), \(j\in\mathcal{P}\) and \(\epsilon\in(0,0.5)\), the proposal probabilities \(\bar{r}_j^{(k,m)}\) of the explanatory variables~\(X_j\) converge (in probability) to the respective posterior inclusion probabilities \(\pi_j=\pi(j\in S\,|\,\mathcal{D})\), i.e. for all \(j\in\mathcal{P}\) and \(k=1,\dots,K\) it holds that 
\( \bar{r}_j^{(k,m)} \overset{\text{P}}{\rightarrow} \pi_j \) as \(m\rightarrow\infty\).
\end{corollary}

Thus, the same convergence results hold for the parallel version as for the serial version of MAdaSub. The benefit of the parallel algorithm is that the convergence of the proposal probabilities against the posterior inclusion probabilities can be accelerated via the exchange of information between the parallel chains, so that the MCMC chains can converge faster against the full posterior distribution. There is a practical trade-off between the effectiveness regarding the joint update for the proposal probabilities and the efficiency regarding the communication between the different chains. If the number of rounds~\(R\) is chosen to be small with a large number of iterations~\(T\) per round, the available information from the multiple chains is not fully utilized during the algorithm; however, if the number of rounds~\(R\) is chosen to be large with a small number of iterations~\(T\) per round, then the computational cost of communication between the chains increases and may outweigh the benefit of the accelerated convergence of the proposal probabilities. 
If~\(T_{\text{max}}\) denotes the maximum number of iterations, we observe that choosing the number of rounds \(R\in[10,100]\) with \(T=T_{\text{max}}/R\) iterations per round works well in practice (see Sections~\ref{sec:sim} and~\ref{sec:realdata} as well as Table~G.4 of the Supplement). 

\section{Simulated data applications}\label{sec:sim}

\subsection{Illustrative example}
 
We first illustrate the adaptive behaviour of the serial MAdaSub algorithm (Algorithm~\ref{algo:MCMC}) in a relatively low-dimensional setting. 
In particular, we consider an illustrative simulated dataset \(\mathcal{D}=(\bs X,\bs y)\) with sample size \(n=60\) and \(p=20\) explanatory variables, by generating \( \bs X=(X_{i,j})\in\mathbb{R}^{n\times p}\) with \(i\)-th row \(\bs X_{i,*}\sim\mathcal{N}_p(\bs 0,\bs \Sigma)\), where \(\bs \Sigma=(\Sigma_{i,j})\in\R^{p \times p}\) is the covariance matrix with entries \(\Sigma_{k,l}=\rho^{|k-l|}\), \(k,l\in\{1,\dots,p\}\), corresponding to a Toeplitz correlation structure with \(\rho=0.9\).  The true vector of regression coefficients is considered to be \[\bs \beta_0=(0.4,0.8,1.2,1.6,2.0,0,\dots,0)^T\in\mathbb{R}^p\,,\] with active set \(S_0=\{1,\dots,5\}\). The response \(\bs y=(y_1,\dots,y_n)^T\) is then simulated from the normal linear model via \(y_i\overset{\text{ind.}}{\sim} N(\bs X_{i,*} \bs \beta_0, 1)\), \(i=1,\dots,n\). 
We employ the g-prior with \(g=n\) and an independent Bernoulli model prior with inclusion probability \(\omega=0.5\), resulting in a uniform prior over the model space (see Remark~\ref{remark:conjugate}). In the MAdaSub algorithm we set \(r_j^{(0)}=\frac{1}{2}\) for \(j\in\mathcal{P}\), i.e.\ we use the prior inclusion probabilities as initial proposal probabilities.  
We first consider the choice \(L_j=p\) (for \(j\in\mathcal{P}\)) for the variance parameters of MAdaSub, corresponding to  equation~(\ref{eq:update_no}). Furthermore, we set \(\epsilon=\frac{1}{p}\) and run the MAdaSub algorithm for \(T=20{,}000\) iterations. To compare the results of MAdaSub with the true posterior model distribution, we have also conducted a full model enumeration using the Bayesian Adaptive Sampling~(BAS) algorithm, which is implemented in the R-package \texttt{BAS}~\citep{clyde2017}.

\begin{figure}[!t]\centering
\includegraphics[width=\textwidth]{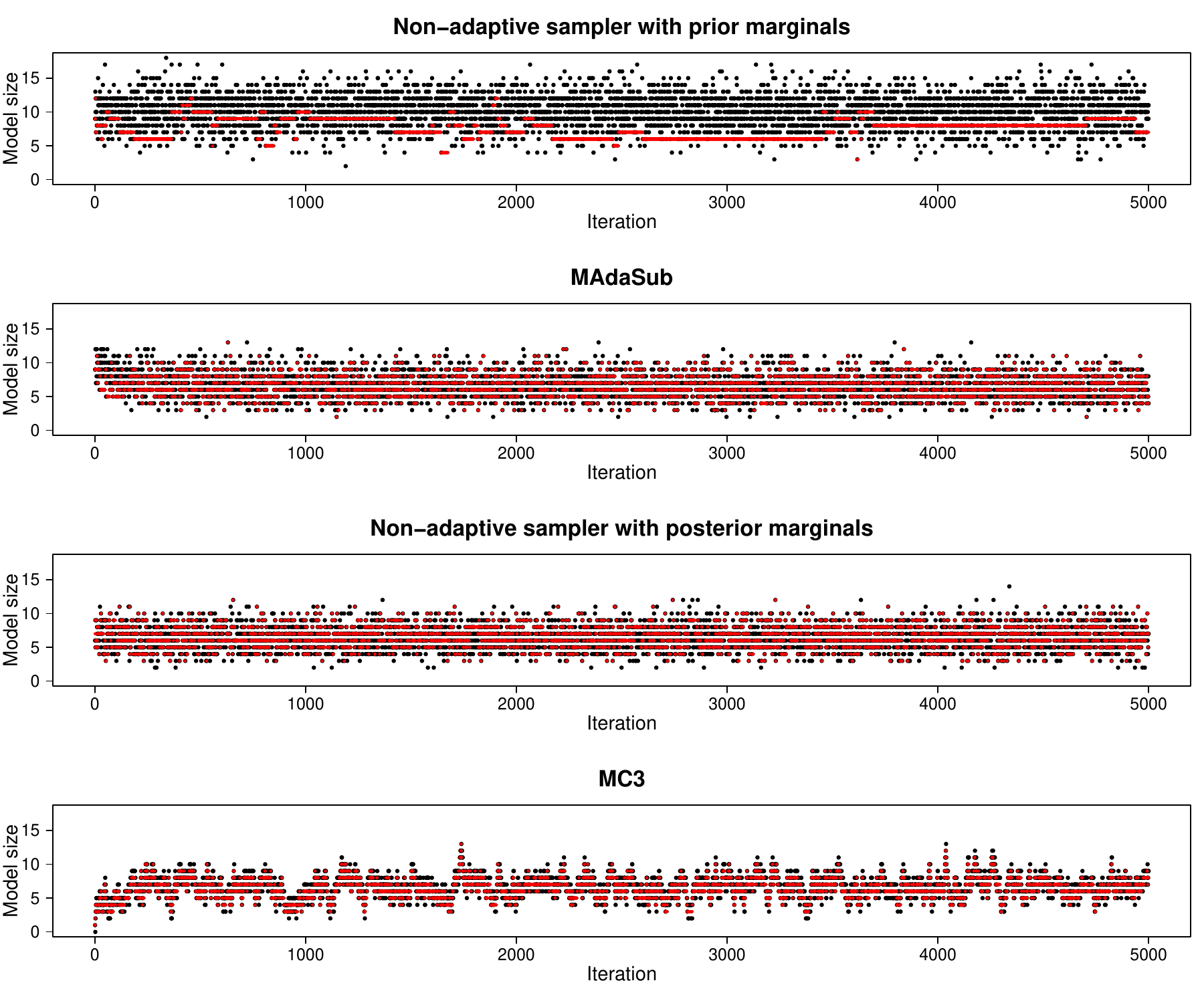} 
\caption{\label{fig:lowdim_sizes}Illustrative example with g-prior. \small Evolution of the sizes \(|V^{(t)}|\) of the proposed models (black) and of the sizes \(|S^{(t)}|\) of the sampled models (red) along the first 5,000 iterations~(\(t\)) for non-adaptive sampler with prior marginals as proposal probabilities
, for MAdaSub (with \(L_j=p\)), for non-adaptive sampler with posterior marginals as proposal probabilities and for local add-delete \(\text{MC}^3\) sampler (from top to bottom).} 
\end{figure}

To illustrate the efficient adaptation of MAdaSub, we present comparisons with independent Metropolis-Hastings algorithms where the individual proposal probabilities are \textit{not} adapted during the algorithm, i.e.\ we set \(r_j^{(t)}=r_j^{(0)}\) for all \(t\in\N\) and \(j\in\mathcal{P}\). In particular, we consider the choice \(r_j^{(t)}=r_j^{(0)}=0.5\), corresponding to the initial proposal distribution in MAdaSub, and the choice \(r_j^{(t)}=r_j^{(0)}=\pi(j\in S\,|\,\mathcal{D})\), corresponding to the targeted proposal distribution, which is, as stated above, the closest independent Bernoulli proposal to the target \(\pi(\cdot\,|\,\mathcal{D})\) in terms of Kullback-Leibler divergence (\citealp{clyde2011}). Note that the non-adaptive independence sampler with posterior inclusion probabilities as proposal probabilities (\(r_j^{(t)}=\pi(j\in S\,|\,\mathcal{D})\)) is only considered as a benchmark and cannot be used in practice, since the true posterior probabilities are initially unknown and are to be estimated by the MCMC algorithms. Furthermore, we also present comparisons with a standard local ``Markov chain Monte Carlo model composition'' (\(\text{MC}^3\)) algorithm \citep{madigan1995}, which in each iteration proposes to delete or add a single variable to the current model. 

\begin{figure}[!t]\centering
\includegraphics[width=0.9\textwidth]{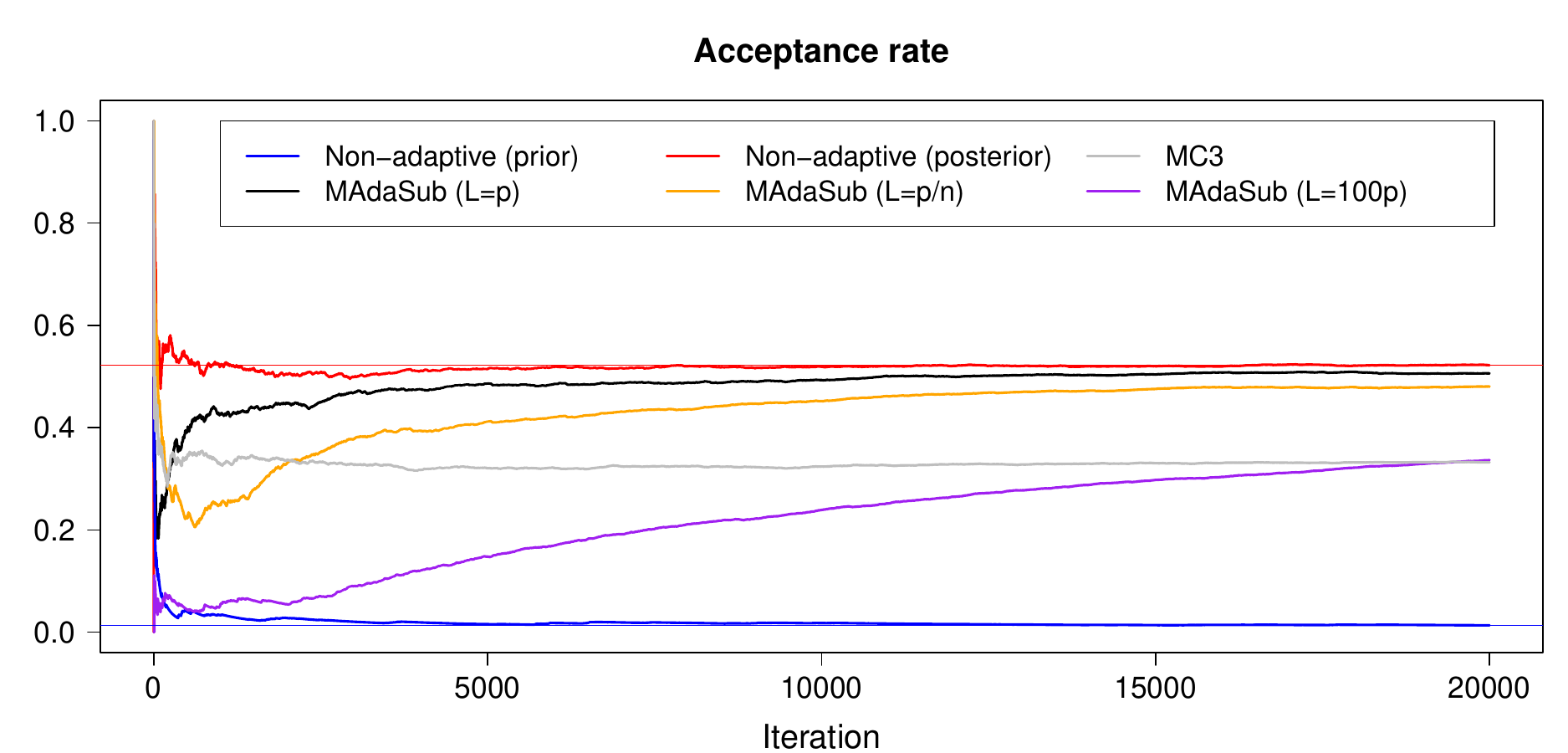}
\caption{\label{fig:lowdim_acc} Illustrative example with g-prior. \small Evolution of acceptance rates along the iterations for non-adaptive independence sampler with prior marginals (blue) and posterior marginals (red) as proposal probabilities, for add-delete \(\text{MC}^3\) sampler (gray), as well as for MAdaSub with \(L_j=p\) (black), \(L_j=p/n\) (orange) and \(L_j=100p\) (purple) for \(j\in\mathcal{P}\).} 
\end{figure}
 
Figure~\ref{fig:lowdim_sizes} depicts the sizes \(|V^{(t)}|\) of the proposed models and the sizes \(|S^{(t)}|\) of the sampled models, while Figure~\ref{fig:lowdim_acc} shows the evolution of the acceptance rates along the iterations \(t\) of the different MCMC algorithms. 
As might have been expected, the non-adaptive sampler with prior marginals as proposal probabilities performs poorly with a very slow exploration of the model space and a small acceptance rate which remains close to zero. On the other hand, the non-adaptive sampler with posterior marginals as proposal probabilities leads to fast mixing with corresponding acceptance rate of approximately \(0.54\). Even though the MAdaSub algorithm starts with exactly the same ``initial configuration'' as the non-adaptive sampler with prior marginals, it quickly adjusts the proposal probabilities accordingly, so that the resulting acceptance rate approaches the target value of \(0.54\) from the non-adaptive sampler with posterior marginals. In particular, when inspecting the evolution of the sampled model sizes in Figure~\ref{fig:lowdim_sizes}, the MAdaSub algorithm is very difficult to distinguish from the sampler with posterior marginals after a very short burn-in period (see also Figure~E.1 of the Supplement).

To illustrate the behaviour of the MAdaSub algorithm with respect to the variance parameters~\(L_j\), additionally to the choice \(L_j=p\) we examine two further runs of MAdaSub with the same specifications as before, but with \(L_j=p/n\) 
and with \(L_j=100p\), respectively. Figure~\ref{fig:lowdim_acc} indicates that the original choice \(L_j=p\) is favourable, yielding a fast and ``sustainable'' increase of the acceptance rate (see also Figure~E.2 of the Supplement for the evolution of proposal probabilities for the different~\(L_j\)). On the other hand, for \(L_j=100p\) the proposal probabilities in MAdaSub are slowly adapted, while for \(L_j=p/n\) the proposal probabilities are adapted very quickly, resulting in initially large acceptance rates; however, this increase is only due to a premature focus of the proposal on certain parts of the model space and thus the acceptance rate decreases at some point when the algorithm identifies other areas of high posterior probability that have not been covered by the proposal. 
This illustrative example shows that --- despite the ergodicity of the MAdaSub algorithm for all choices of its tuning parameters (Theorem~\ref{thm:MAdaSub}) --- the speed of convergence against the target distribution crucially depends on an appropriate choice of these parameters. Regarding the variance parameters we observe that the choice~\(L_j=p\) for~\(j\in\mathcal{P}\) works well in practice (see also results below).

The adaptive nature of MAdaSub entails the possibility for an automatic check of convergence of the algorithm: as the proposal probabilities~\(r_j^{(t)}\) are continuously adjusted towards the current empirical inclusion frequencies \(f_j^{(t)}=\frac{1}{t}\sum_{i=1}^t \mathbbm{1}_{S^{(i)}}(j)\) (see equation~(\ref{eq:update})), the algorithm may be stopped as soon as the individual proposal probabilities and empirical inclusion frequencies are within a prespecified distance \(\delta\in(0,1)\) (e.g.~\(\delta=0.005\), see~Figure~E.3 of the Supplement), i.e.\ the algorithm is stopped at iteration~\(t_c\) if \(\max_{j\in\mathcal{P}}|f_j^{(t_c)}-r_j^{(t_c)}|\leq \delta\). Even when automatic stopping may be applied, we additionally recommend to investigate the convergence of the MAdaSub algorithm via the diagnostic plots presented in this section and in Section~E of the Supplement. 

\subsection{Low-dimensional simulation study}

In this simulation study we further investigate the performance of the serial MAdaSub algorithm in relation to local non-adaptive and adaptive algorithms. In particular, we analyse how the algorithms are affected by high correlations between the covariates. 

We consider a similar low-dimensional setting as in the illustrative data application with \(p=20\) covariates and sample size \(n=60\). To evaluate the performance in a variety of different data situations, for each simulated dataset the number \(s_0\) of informative variables is randomly drawn from \(\{0,1,\dots,10\}\) and the true active set \(S_0\subseteq\mathcal{P}\) of size \(|S_0|=s_0\) is randomly selected from the full set of covariates \(\mathcal{P}=\{1,\dots,p\}\); then, for each \(j\in S_0\), the \(j\)-th component \(\beta_{0,j}\) of the true coefficient vector \(\bs\beta_0\in\R^p\) is simulated from a uniform distribution \(\beta_{0,j}\sim U(-2,2)\). As before, the covariates are simulated using a Toeplitz correlation structure, while the response is simulated from a normal linear model with error variance \(\sigma^2=1\).  We consider three different correlation settings by varying the correlation \(\rho\) between adjacent covariates in the Toeplitz structure:\ a low-correlated setting with \(\rho=0.3\), a highly-correlated setting with \(\rho=0.9\) and a very highly-correlated setting with \(\rho=0.99\). For each of the three settings, 200 different datasets are simulated as described above; in each case, we employ a g-prior with \(g=n\) on the regression coefficients and a uniform prior on the model space. 

\begin{figure}[!t]\centering 
\includegraphics[width=\textwidth]{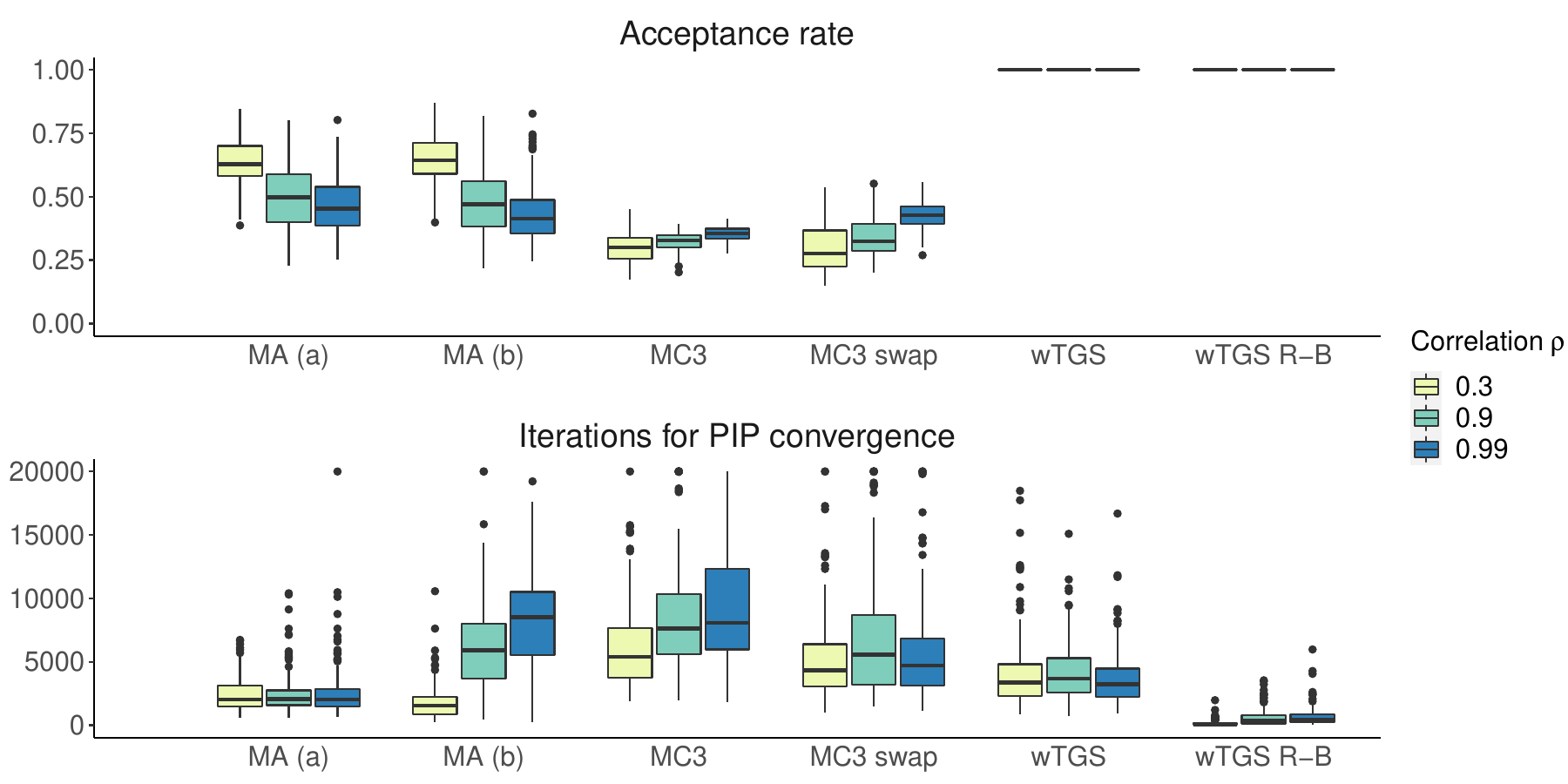}
\caption{\label{fig:simulations_corr} Low-dimensional simulation study with varying correlation \(\rho\in\{0.3,0.9,0.99\}\) in Toeplitz structure. \small Performance of MAdaSub with initial proposal probabilities \(r_j^{(0)}=0.5\) based on prior inclusion probabilities~(MA(a)), MAdaSub with \(r_j^{(0)}\) being based on marginal posterior odds*~(MA(b)), \(\text{MC}^3\) samplers with and without ``swap'' moves, as well as adaptive weighted Tempered Gibbs Sampler based on weighted frequency estimates~(wTGS) and Rao-Blackwellized estimates~(wTGS R-B), in terms of acceptance rates and numbers of iterations for convergence of the estimates to the true posterior inclusion probabilities~(PIP). \\[1mm]
 *The (approximated) marginal posterior odds are provided in equation~(\ref{eq:margodds}). }  
\end{figure}

For each simulated dataset we apply MAdaSub with 20,000 iterations, using \(L_j=p\) for \(j\in\mathcal{P}\) and \(\epsilon=\frac{1}{p}\). In order to investigate the influence of the initial proposal probabilities \(r_j^{(0)}\) in MAdaSub, two different choices for \(r_j^{(0)}\) are considered: choice~(a)~based on prior inclusion probabilities \(r_j^{(0)}=\frac{1}{2}\) and choice~(b)~based on (approximated) marginal posterior odds  
\begin{equation} \pi_j^{\text{marg}} = \frac{\text{PO}_j}{1+\text{PO}_j} ~~~ \text{with } ~~ \text{PO}_j = \frac{P( S=\{j\}\,|\,\mathcal{D})}{P(S=\emptyset\,|\,\mathcal{D})},~ j\in\mathcal{P} , \label{eq:margodds} \end{equation}
and setting \(r_j^{(0)} = \min \{ \max\{ \pi_j^{\text{marg}}, \frac{1}{p} \} , 0.9 \} \) to prevent the premature focus of the algorithm on some covariates (if \(\pi_j^{\text{marg}}\approx 1\)) or the avoidance of other covariates  (if \(\pi^{\text{marg}}_j\approx 0\)). Here, the marginal posterior odds \(\text{PO}_j\) are crude approximations to the true posterior odds, derived under the assumption of posterior independence of variable inclusion. The local \(\text{MC}^3\) algorithm \citep{madigan1995} is applied as before as well as with additional swap moves to potentially improve the mixing (as in~\citealp{griffin2018}). Using the R-package~\texttt{scaleBVS}~\citep{scaleBVS2020}, we apply the adaptive weighted tempered Gibbs sampling algorithm of~\cite{zanella2019} to obtain (weighted) frequency estimates 
(as for the other algorithms) and Rao-Blackwellized estimates of posterior inclusion probabilities~(PIPs). Exact PIPs are again derived using the BAS algorithm \citep{clyde2017}. The algorithms are evaluated based on final acceptance rates and numbers of iterations for convergence of the estimates~$\hat{f}_j^{(t)}$ to the true PIPs, where PIP convergence is defined to occur at the smallest iteration~\(t_c\) for which \(\max_{j\in\mathcal{P}}|\hat{f}_j^{(t_c)}-\pi_j|\leq 0.05\); if \(t_c\geq20{,}000\), then the number of iterations for convergence is displayed as 20,000 in Figure~\ref{fig:simulations_corr}. 

Figure~\ref{fig:simulations_corr} shows that the acceptance rates of the MAdaSub samplers tend to be substantially larger in comparison to the local \(\text{MC}^3\) algorithms, while the acceptance rates of the weighted Tempered Gibbs Sampler~(wTGS) are equal to one by construction. Nevertheless, for the MAdaSub samplers a decreasing trend of acceptance rates can be observed with increasing correlations. This observation reflects that for low-correlated situations the resulting posterior distribution is often closer to an independent Bernoulli form than for highly-correlated cases, and thus can be better approximated by the proposal distributions of MAdaSub, leading to larger acceptance rates. In the low-correlated setting (\(\rho=0.3\)), the choice~(b) for the initial proposal probabilities in MAdaSub based on marginal posterior odds leads to slightly larger acceptance rates and a faster PIP convergence compared to the MAdaSub sampler~(a) based on the prior inclusion probabilities. However, in cases of high correlations among some of the covariates (\(\rho=0.9\) and \(\rho=0.99\)), the prior choice~(a) is clearly favourable yielding larger acceptance rates and a faster PIP convergence compared to the MAdaSub sampler~(b) and the \(\text{MC}^3\) algorithm. Thus, while in low-correlated settings the marginal posterior odds yield reasonable first approximations to the true posterior odds, the prior inclusion probabilities are more robust and to be preferred as initial proposal probabilities in MAdaSub in situations with high correlations. Overall, the MAdaSub sampler~(a) yields a well-mixing algorithm in all considered settings, which is also competitive to the adaptive wTGS algorithm based on weighted frequency estimates, while wTGS with Rao-Blackwellization~(R-B) provides faster convergence. 
Note that the computational cost of R-B is small in this low-dimensional conjugate setting but increases for high-dimensional and non-conjugate settings with Laplace approximations~\citep{zanella2019, wan2021}. 
An additional sensitivity analysis regarding different variance parameters~$L_j$ in MAdaSub (see Figure~F.1 of the Supplement) supports the choice~$L_j=p=20$ in all considered correlation settings and indicates that results are very robust for~$L_j\in[p/2,2p]$.

\subsection{High-dimensional simulation study}\label{sec:highdim}

To investigate the performance of the serial and parallel versions of MAdaSub in high-dimensional settings, we consider the same simulation set-up as in~\cite{yang2016} and~\cite{griffin2018}: data are simulated from a sparse linear regression model with true coefficients 
\begin{equation}\bs\beta_0= \text{SNR} \times \sqrt{\log(p)/n} \times (2, -3, 2, 2, -3, 3, -2, 3, -2, 3, 0, \dots, 0)^T\in \R^p.\end{equation} 
Similar to the low-dimensional simulations, covariates are generated from a Toeplitz correlation structure with~\(\rho=0.6\) and the response is simulated via \(y_i\overset{\text{ind.}}{\sim} N(\bs X_{i,*} \bs \beta_0, 1)\), \(i=1,\dots,n\). 
As in \cite{griffin2018}, we consider the conjugate prior~(\ref{eq:prior1}) with \(g=9\) and prior independence of the regression coefficients (\(\bs W_S = \bs I_{|S|}\) for \(S\in\mathcal{M}\)), 
together with the model prior~(\ref{eq:modelprior}) with (fixed) prior inclusion probability \(\omega=10/p\). 
For each setting with~\(n\in\{500,1000\}\), \(p\in\{500,5000\}\) and signal-to-noise ratio~\(\text{SNR}\in\{0.5,1,2,3\}\), we simulate one dataset and apply each algorithm 200 times to assess the stability of estimated posterior inclusion probabilities. As in~\cite{griffin2018}, each algorithm is based on 5 parallel chains using 5 CPUs. 
We consider the serial version of MAdaSub where the individual chains (Algorithm~\ref{algo:MCMC}) are run in parallel but do not exchange any information and the parallel version (Algorithm~2 of the Supplement) where the chains exchange information regarding the proposal probabilities after each of $R=50$ rounds (considering 25 burn-in rounds for both versions; each round consists of 1000 and 10,000 iterations for \(p=500\) and \(p=5000\), respectively). For the serial version, the initial proposal probabilities are set to the prior inclusion probabilities, i.e.\ \(r_j^{(k,0)}=10/p\), and the variance parameters \(L_j^{(k)}=p\) are the same for all chains~\(k\). For the parallel version, we consider different random initializations of proposal probabilities  \(r_j^{(k,0)}=q^{(k)}/p\), \(j\in\mathcal{P}\), with \(q^{(k)}\sim U(2,10)\) and variance parameters \(L_j^{(k)}=L^{(k)}\), \(j\in\mathcal{P}\), with \(L^{(k)}\sim U(p/2,2p)\) for each chain~\(k\). For all MAdaSub chains we set~\(\epsilon=1/p\). Additional results of sensitivity analyses regarding different choices of the tuning parameters of MAdaSub can be found in Section~G of the Supplement. 

The performance of the MAdaSub algorithms~(A) with serial and parallel updating schemes is assessed in terms of median acceptance rates, as well as in comparison to the add-delete-swap \(\text{MC}^3\) algorithm~(B) in terms of the median estimated ratio~$\hat{r}_{A,B}^{(20)}$ of the relative time-standardized effective sample size of algorithm~$A$ versus algorithm~$B$ for the posterior inclusion probabilities~(PIPs) over the 20 variables with the largest estimated PIPs (averaged over all algorithms). The estimated ratio of the relative time-standardized effective sample size is given by $\hat{r}_{A,B}=(s_B^2t_B)/(s_A^2t_A)$, with $t_A$ and $t_B$ the median computation times and $s_A^2$ and $s_B^2$ the variances of PIP estimates based on 200 independent runs of each algorithm (cf.~\citealp{griffin2018}). Here, we consider the median ratio~$\hat{r}_{A,B}^{(20)}$ over the 20 variables with the largest estimated PIPs, as many variables receive very small posterior probability due to the sparsity-inducing prior and the sparse generating model with only 10 signal variables (in all settings the estimated PIPs for variables not among the top 20 are all below 0.5\%, while the median estimated PIP over all variables is below 0.07\%). Complimentary results regarding the median of~$\hat{r}_{A,B}$ over all variables are provided in Table~G.1 of the Supplement, comparing the performance of MAdaSub also with the adaptive approaches in~\cite{griffin2018}. 

\begin{table*}
\begin{center}
\resizebox{\textwidth}{!}{\begin{tabular}{ clR{2.4cm}R{2.4cm}R{2.4cm}R{2.4cm} } 
 \hline
  & & \multicolumn{1}{c}{$\text{SNR}=0.5$} & \multicolumn{1}{c}{$\text{SNR}=1$} & \multicolumn{1}{c}{$\text{SNR}=2$} & \multicolumn{1}{c}{$\text{SNR}=3$}  \\ 
 $(n,p)$ & MAdaSub & $\hat{r}_{A,B}^{(20)}$ / \,\,\,Acc.  &  $\hat{r}_{A,B}^{(20)}$ / \,\,\,Acc. &  $\hat{r}_{A,B}^{(20)}$ / \,\,\,Acc. & $\hat{r}_{A,B}^{(20)}$ / \,\,\,Acc.  \\ \hline \hline 
 $ (500,500)$        & serial & 69.4 / 44.6\% & 23.0 / 31.9\%  & 4.8 / \,\,\,6.3\%  & 8.3  / \,\,\,9.3\% \\
                     & parallel & 22.9 / 45.3\% & 8.9  / 37.7\%  & 7.5 / 18.1\% & 12.1 / 21.4\% \\ 			\hline
 $ (500,5000)$       & serial & 376.9 / 47.5\% & 50.3 / 46.6\% & 8.2 / \,\,\,5.1\% & 17.9  / \,\,\,9.5\% \\
                     & parallel & 474.4 / 48.0\% & 78.7 / 44.8\% & 82.8 / 17.5\% & 186.4 / 23.4\% \\
										\hline
	 $ (1000,500)$     & serial & 110.7 / 53.4\% & 13.7 / 39.0\%  & 2.4 / \,\,\,6.0\%  & 8.7 / \,\,\,9.0\% \\
                     & parallel & 62.0  / 54.2\% & 7.0 / 39.0\%   & 7.3 / 17.7\% & 12.8 / 21.0\% \\ 			\hline
	 $ (1000,5000)$    & serial & 657.3 / 45.3\% & 7.5 / 26.5\%  & 23.9  / \,\,\,9.4\%  & 35.1  / 11.6\%  \\
                     & parallel & 674.1 / 45.8\% & 6.2  / 10.7\%  & 175.6 / 23.1\% & 281.7 / 24.7\% \\		\hline							
\end{tabular}}
\end{center}
\caption{Results of high-dimensional simulation study. \small Performance of MAdaSub algorithms~(A) with serial and parallel updating schemes compared to add-delete-swap \(\text{MC}^3\) algorithm~(B) in terms of median estimated ratios~$\hat{r}_{A,B}^{(20)}$ of the relative time-standardized effective sample size for PIPs over the 20 variables with the largest estimated PIPs. Median acceptance rates (Acc.) for MAdaSub are also provided. 
}
\label{tab:highdim}
\end{table*}

Table~\ref{tab:highdim} shows that in all considered settings the median estimated time-standardized effective sample size for both MAdaSub versions is several orders larger than for the \(\text{MC}^3\) algorithm. For low SNRs (e.g.\ \(\text{SNR}=0.5\)), both MAdaSub versions tend to show larger improvements compared to the \(\text{MC}^3\) algorithm than for high SNRs (e.g.\ \(\text{SNR}=3\)). Note that for high SNRs, the posterior distribution tends to be more concentrated around the true model \(S_0=\{1,\dots,10\}\), so that local proposals of the add-delete-swap \(\text{MC}^3\) algorithm may also be reasonable. On the other hand, for low SNR, the posterior tends to be less concentrated, so that global moves of MAdaSub have a larger potential to improve the mixing compared to the \(\text{MC}^3\) algorithm. The acceptance rates of MAdaSub are also larger in small SNR scenarios, as the posterior model distribution tends to be better approximated by independent Bernoulli proposals. However, in all considered settings, the acceptance rates of MAdaSub are reasonably large with median acceptance rates between 5.1\% and 54.2\% (see Table~\ref{tab:highdim}) and are considerably larger compared to the \(\text{MC}^3\) algorithm with median acceptance rates between 0.6\% and 5.8\% (detailed results not shown).

For low SNRs ($\text{SNR}\leq1$), serial updating in MAdaSub tends to yield larger (for \(p=500\)) or similar (for \(p=5000\)) time-standardized effective sample sizes compared to parallel updating, as both versions appear to have converged to stationarity with similar acceptance rates, while the parallel version tends to yield larger computation times as a result of communicating chains. For large SNRs ($\text{SNR}\geq2$), MAdaSub with parallel updating performs favourable since the proposal probabilities tend to converge faster than with serial updating, which leads to considerably larger acceptance rates and outweighs the computational cost of communicating chains. 
Previous results for the same simulation set-up indicate that the two alternative individual adaptation algorithms of~\cite{griffin2018} tend to yield the largest improvements compared to the \(\text{MC}^3\) algorithm for higher SNR (particularly for \(\text{SNR}=2\)). The proposal~\eqref{eq:propalt} of these algorithms allows for larger moves than the add-delete-swap proposal in \(\text{MC}^3\), but --- in contrast to the independence proposal of MAdaSub --- the proposal~\eqref{eq:propalt} still locally depends on the previously sampled model. Overall, MAdaSub shows a competitive performance compared to the adaptive algorithms of~\cite{griffin2018}, with advantages of MAdaSub in low SNR settings and advantages of the adaptive algorithms of~\cite{griffin2018} in high SNR settings~(see Table~G.1 of the Supplement).

\section{Real data applications}\label{sec:realdata}

\subsection{Tecator data}\label{sec:tecator}

We first examine the Tecator dataset which has already been investigated in \citet{griffin2010}, \citet{lamnisos2013} and \citet{griffin2018}. The data has been recorded by \citet{borggaard1992} on a Tecator Infratec Food Analyzer and consists of \(n=172\) meat samples and their near-infrared absorbance spectra, represented by \(p=100\) channels in the wavelength range 850-1050nm (compare \citealp{griffin2010}). The fat content of the samples is considered as the response variable. 
 For comparison reasons, we choose the same conjugate prior set-up as in \citet{lamnisos2013}, i.e.\ we use the prior given in equation~(\ref{eq:prior1}) with \(g=5\), \(\bs W_S = \bs I_{|S|}\) for \(S\in\mathcal{M}\) 
and we employ the independent Bernoulli model prior given in equation~(\ref{eq:modelprior}) with (fixed) prior inclusion probability \(\omega=\frac{5}{100}\). 

\begin{figure}[!t]\centering 
\includegraphics[width=\textwidth]{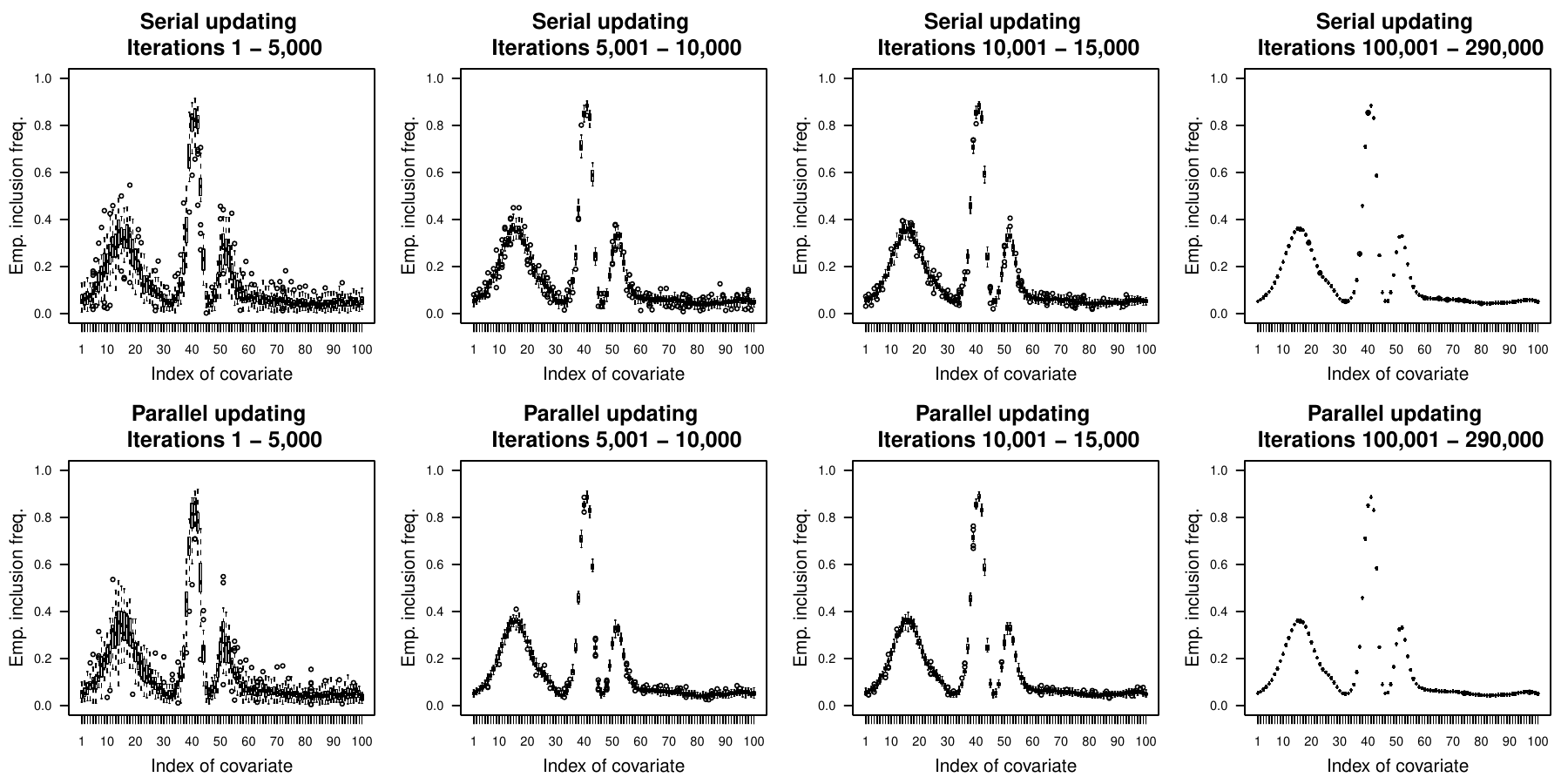}
\caption{\label{fig:Tecator_parallel}Tecator data application. \small Results of 25 independent serial MAdaSub chains (Algorithm~\ref{algo:MCMC}) and of 25 parallel MAdaSub chains exchanging information after every 5,000 iterations (Algorithm~2) in terms of empirical variable inclusion frequencies~\(f_j\) for \(j\in\{1,\dots,100\}\).} 
\end{figure}

To investigate the stability of MAdaSub for different choices of its tuning parameters, we run 25 independent serial MAdaSub chains (Algorithm~\ref{algo:MCMC}) 
with random initializations of the proposal probabilities \(r_j^{(k,0)}=q^{(k)}/p\), \(j\in\mathcal{P}\), with \(q^{(k)}\sim U(2,10)\) and of the variance parameters \(L_j^{(k)}=L^{(k)}\), \(j\in\mathcal{P}\), with \(L^{(k)}\sim U(p/2,2p)\), for each chain \(k=1,\dots,25\). Furthermore, we run 25 additional parallel MAdaSub chains (Algorithm~2) with the described random initializations, exchanging the information after each of \(R=58\) rounds of \(T=5{,}000\) iterations (yielding in total 290,000 iterations for each of the chains, cf.~\citealp{lamnisos2013}). 
Figure~\ref{fig:Tecator_parallel} shows the resulting empirical variable inclusion frequencies (as estimates of posterior inclusion probabilities) for the 25 serial and 25 parallel MAdaSub chains. From left to right, the first three plots of Figure~\ref{fig:Tecator_parallel} depict the development of the empirical inclusion frequencies for the first three rounds of 5,000 iterations each, while the rightmost plots depict the final empirical inclusion frequencies after 290,000 iterations (disregarding a burn-in period of 100,000 iterations, cf.~\citealp{lamnisos2013}).  After the first 5,000 iterations, the empirical inclusion frequencies show a similar variability for the serial and parallel chains, as no communication between the parallel chains has yet occurred. After the second round of 5,000 further iterations, the benefit of the communication between the 25 parallel chains is apparent, leading to less variable estimates due to a faster convergence of the proposal probabilities against the posterior inclusion probabilities. Nevertheless, also the serial MAdaSub chains (with different initial tuning parameters) provide quite accurate estimates after only 10,000 iterations. 

After 290,000 iterations, all of the serial and parallel MAdaSub chains yield very stable estimates of posterior inclusion probabilities, reproducing the results shown in Figure~1 of \citet{lamnisos2013}. Details on additional comparisons with \citet{lamnisos2013} and computation times can be found in Section~H of the Supplement. As the covariates represent 100 channels of the near-infrared absorbance spectrum, adjacent covariates are highly correlated and it is not surprising that they have similar posterior inclusion probabilities. If one is interested in selecting a final single model, the median probability model (which includes all variables with posterior inclusion probability greater than 0.5, see \citealp{barbieri2004}) might not be the best choice in this particular situation, since then only variables corresponding to the ``global mode'' and no variables from the two other ``local modes'' in Figure~\ref{fig:Tecator_parallel} are selected. Alternatively, one may choose one or two variables from each of the three ``local modes'' or make use of Bayesian model averaging \citep{raftery1997} for predictive inference.   

\subsection{PCR and Leukemia data}\label{sec:PCR}

We illustrate the effectiveness of MAdaSub for two further high-dimensional datasets. In particular, we consider the polymerase chain reaction (PCR) dataset of \citet{lan2006} with \(p=22{,}575\) explanatory variables (expression levels of genes), sample size \(n=60\) (mice) and continuous response data (the dataset is available in JRSS(B) Datasets Vol. 77(5), \citealp{Song2015}). Furthermore, we consider the leukemia dataset of \citet{golub1999} with \(6817\) gene expression measurements of \(n=72\) patients and binary response data (the dataset can be loaded via the R-package \texttt{golubEsets}, \citealp{golub2017}). For the PCR dataset we face the problem of variable selection in a linear regression framework, while for the leukemia dataset we consider variable selection in a logistic regression framework. We have preprocessed the leukemia dataset as described in \citet{dudoit2002}, resulting in a restricted design matrix with \(p=3571\) columns (genes). Furthermore, in both datasets we have mean-centered the columns of the design matrix after the initial preprocessing.

Here we adopt the posterior approximation induced by \(\text{EBIC}_\gamma\) with \(\gamma=1\) (see equation~(\ref{eq:EBICkernel})), corresponding to a beta-binomial model prior with \(a_\omega=b_\omega=1\) as parameters in the beta distribution (see Section~\ref{sec:setting}). 
For both datasets we run 25 independent serial MAdaSub chains with 1,000,000 iterations and 25 parallel MAdaSub chains exchanging information after each of \(R=50\) rounds of \(T=20{,}000\) iterations (yielding also 1,000,000 iterations for each parallel chain). For each serial and parallel chain \(k=1,\dots,50\), we set~\(\epsilon=\frac{1}{p}\) and randomly initialize the proposal probabilities \(r_j^{(k,0)}=q^{(k)}/p\), \(j\in\mathcal{P}\), with \(q^{(k)}\sim U(2,5)\) and the variance parameters \(L_j^{(k)}=L^{(k)}\), \(j\in\mathcal{P}\), with \(L^{(k)}\sim U(p/2,2p)\). For the leukemia dataset we make use of a fast C\texttt{++} implementation for ML-estimation in logistic regression models via a limited-memory Broyden-Fletcher-Goldfarb-Shanno (L-BFGS) algorithm, which is available in the R-package \texttt{RcppNumerical} \citep{qiu2016}. For both datasets, the 50 MAdaSub chains are run in parallel on a computer cluster with 50 CPUs, yielding overall computation times of 2,836 seconds for the PCR data (2,310 seconds for a single chain) and 1,402 seconds for the leukemia data (995 seconds for a single chain).

\begin{figure}[!t]\centering 
\includegraphics[width=\textwidth]{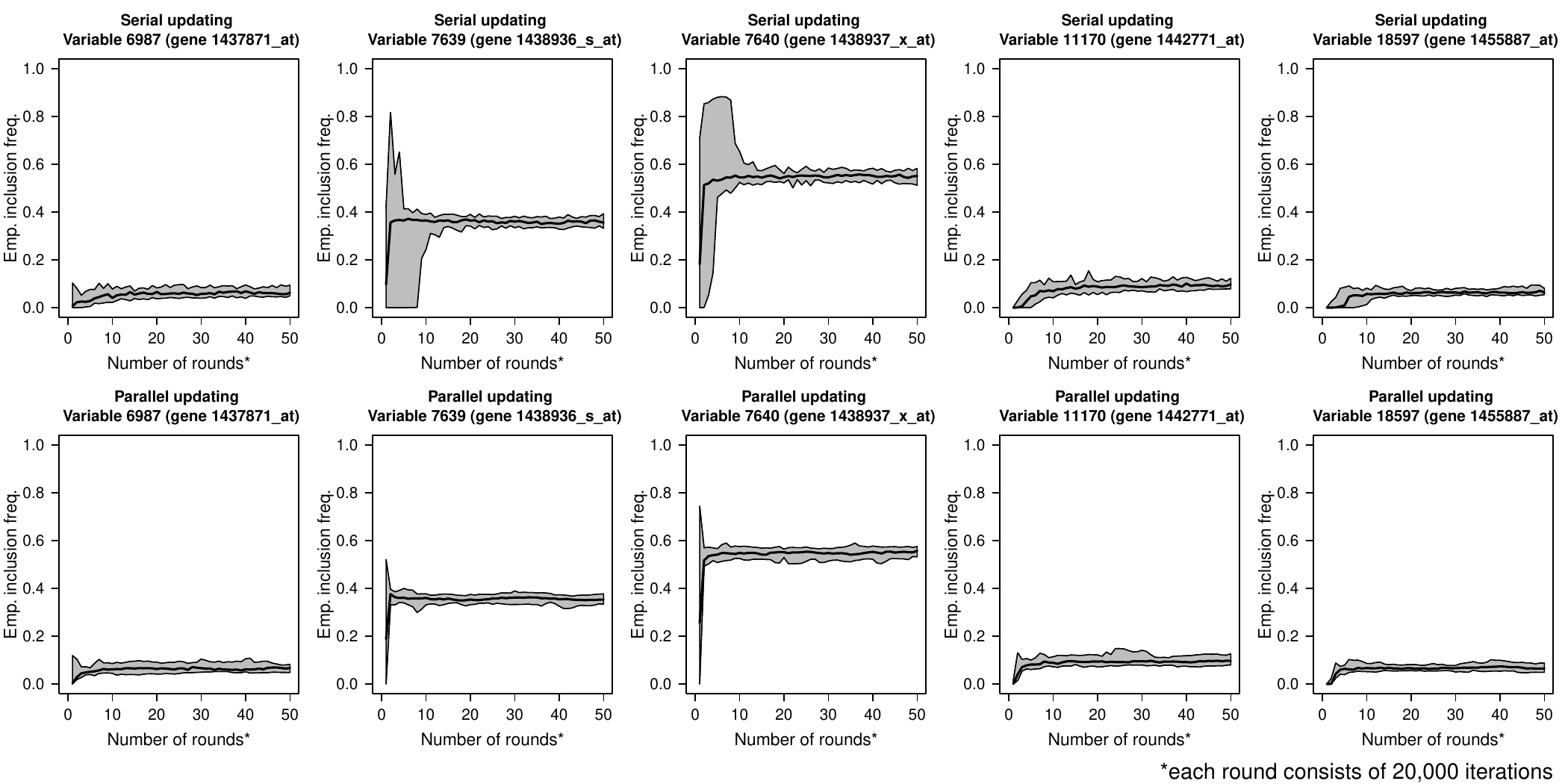}
\caption{\label{fig:PCR_evolution}PCR data application. \small 
Evolution of empirical variable inclusion frequencies 
for 25 serial MAdaSub chains (Algorithm~\ref{algo:MCMC}, top) and 25 parallel MAdaSub chains exchanging information after every round of 20,000 iterations 
(Algorithm~2, bottom). Bold lines represent median frequencies with 5\%- and 95\%-quantiles (shaded area) over the chains within each round, for most informative variables \(X_j\) (with final estimate \(f_j\geq 0.05\) for at least one chain). 
} 
\end{figure}

Figures~\ref{fig:PCR_evolution} and~\ref{fig:Leukemia_evolution} show that, despite the high-dimensional model spaces and the different initializations of each chain, the parallel MAdaSub algorithm provides stable estimates of posterior inclusion probabilities for both datasets after a small number of rounds. 
In particular, the estimates from the parallel MAdaSub algorithm stabilize after only three rounds of 20,000 iterations (see also Figures~I.3 and~I.4 of the Supplement). For the PCR data, all serial and parallel MAdaSub chains yield congruent estimates of posterior inclusion probabilities after 1,000,000 iterations (Figures~\ref{fig:PCR_evolution},~I.2 and~I.3).  
The final acceptance rates of MAdaSub for the PCR dataset are between 20\% and 22\%, while the acceptance rates for the leukemia dataset are between 3\% and 6\%. The smaller acceptance rates for the leukemia dataset indicate that this corresponds to a more challenging scenario (i.e.~the targeted posterior model distribution seems to be ``further away'' from an independent Bernoulli form). This observation is also reflected in the larger variability of the estimates from the MAdaSub chains without parallel updating (Figures~\ref{fig:Leukemia_evolution},~I.2 and~I.4). The leukemia data application particularly illustrates the benefits of the parallel version of MAdaSub, where multiple chains with different initializations sequentially explore different regions of the model space, but exchange the information after each round of 20,000 iterations, increasing the speed of convergence of the proposal probabilities to the posterior inclusion probabilities. 

Note that in very high-dimensional settings such as for the PCR data (with \(p=22{,}575\)), the classical \(\text{MC}^3\) algorithm \citep{madigan1995} does not yield stable estimates due to slow mixing (cf.\ \citealp{griffin2018}), while the BAS algorithm \citep{clyde2017} using sampling without replacement is computationally intractable. Further results in \cite{griffin2018} show that several competing adaptive algorithms --- including sequential Monte Carlo algorithms of \cite{schafer2013} and tempered Gibbs sampling algorithms of \cite{zanella2019} --- do not provide reliable estimates of posterior inclusion probabilities for the PCR data; only the adaptively scaled individual adaptation algorithm of \cite{griffin2018} with proposals of the form (\ref{eq:propalt}) yields stable results for the PCR data similarly to MAdaSub with a slightly different prior set-up (see Figures~10 and~11 of the Supplement of \citealp{griffin2018}).

\begin{figure}[!t]\centering 
\includegraphics[width=\textwidth]{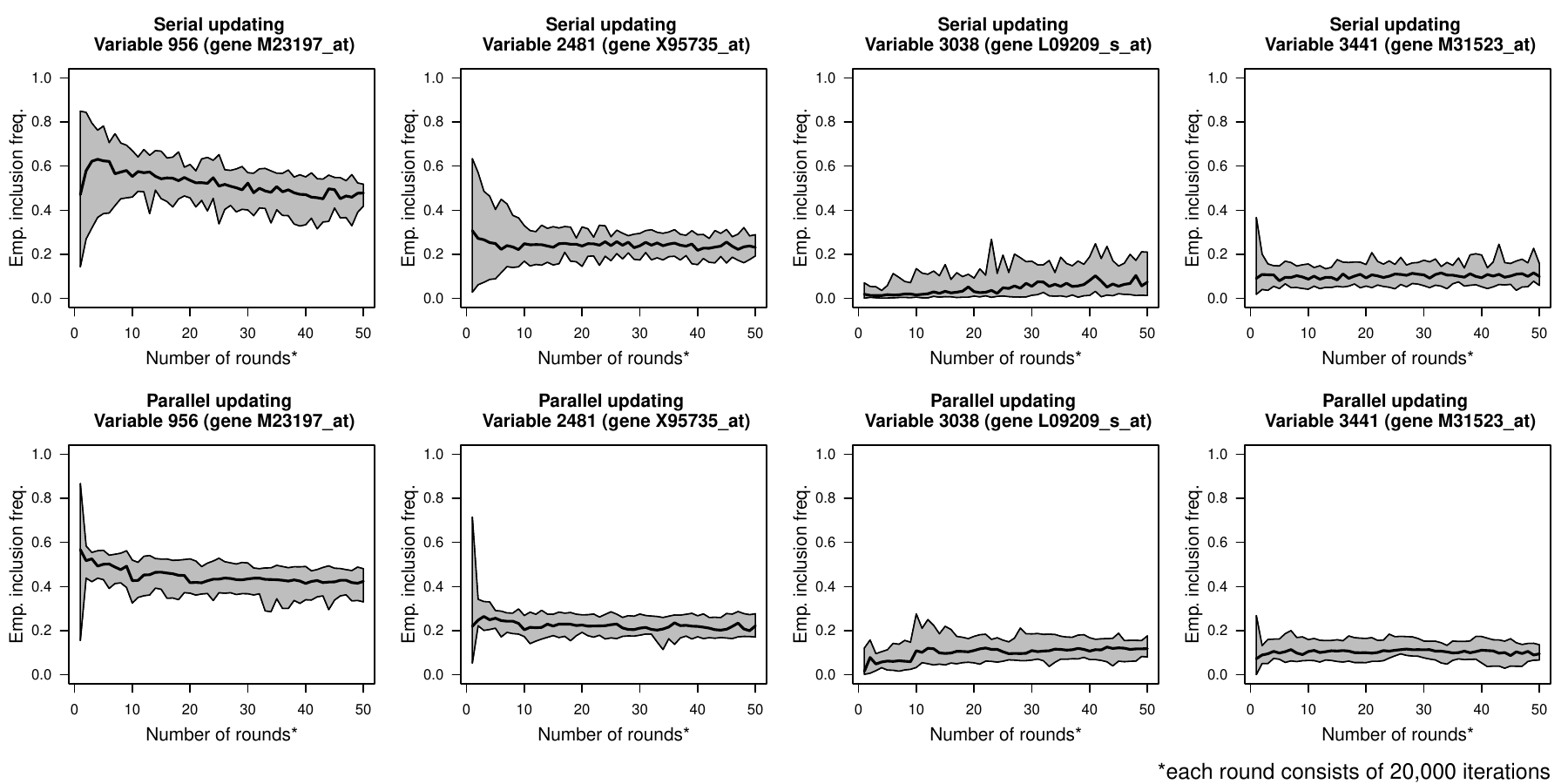}
\caption{\label{fig:Leukemia_evolution}Leukemia data application. \small Evolution of empirical variable inclusion frequencies for 25 serial MAdaSub chains (Algorithm~\ref{algo:MCMC}, top) and 25 parallel MAdaSub chains exchanging information after each round of 20,000 iterations (Algorithm~2, bottom) for most informative variables \(X_j\) (with final estimate \(f_j\geq 0.1\) for at least one chain), cf.\ Figure~\ref{fig:PCR_evolution}.} 
\end{figure}

Due to the very large model spaces in both considered examples, posterior probabilities of individual models are generally small and corresponding MCMC estimates will typically not be very reliable. Therefore, as in similar studies (see \citealp{griffin2018}), we have focused on the estimation of posterior inclusion probabilities~(PIPs). For the PCR data two variables (genes) stand out with respect to the final estimates of their PIPs, namely the gene 1438937\_x\_at (covariate index \(j=7640\)) with estimated PIP between 0.54 and 0.56, and the gene 1438936\_s\_at (\(j=7639\)) with estimated PIP between 0.35 and 0.37. 
Similarly, for the leukemia data two genes stand out, namely the genes M23197\_at (\(j=956\)) with estimated PIP between 0.39 and 0.43 and X95735\_at (\(j=2481\)) with estimated PIP between between 0.21 and 0.22 (considering final estimates from the 25 parallel chains only); these two genes are also among the four top scoring genes in a Bayesian probit regression analysis in \citet{ai2009}.

\section{Discussion}\label{sec:discussion} 
We introduced the Metropolized Adaptive Subspace (MAdaSub) algorithm for sampling from high-dimensional posterior model distributions in situations where conjugate priors or approximations to the posterior are employed. We further developed an efficient parallel version of MAdaSub, where the information regarding the adaptive proposal probabilities of the variables can be shared periodically between the different chains. 
Simulated and real data applications illustrated that MAdaSub can efficiently sample from multimodal posterior model distributions, yielding stable estimates of posterior inclusion probabilities even for ten thousands of possible covariates. 

The reliable estimation of posterior inclusion probabilities is particularly important for Bayesian inference, since the median probability model (MPM) --- including all variables with posterior inclusion probability larger than 0.5 --- has been shown to yield optimal predictions for uncorrelated covariates \citep{barbieri2004} and also a favourable performance for correlated designs \citep{barbieri2021}, e.g.\ compared to the largest posterior probability model. MAdaSub provides a natural adaptive MCMC algorithm which focuses on the sequential adaptation of currently estimated inclusion probabilities, with the aim of driving the sampler quickly into regions near to the MPM; in the limit, the MPM itself is the model which receives the largest probability under the independent Bernoulli proposal of MAdaSub. 
Despite the continuing adaptation of the proposals, we have shown that MAdaSub constitutes a valid MCMC algorithm which samples from the full posterior model distribution. 
While the serial and parallel versions of MAdaSub are ergodic for all choices of their tuning parameters (see Theorem~\ref{thm:MAdaSub} and Theorem~\ref{thm:parallel}), in practice the speed of convergence against the targeted posterior depends crucially on a proper choice of their tuning parameters (see Section~\ref{sec:sim}). Deriving theoretical results regarding the mixing time of the proposed algorithms is an important but challenging issue for further research. 

Since MAdaSub is based on adaptive independent proposal distributions, in each iteration of the algorithm the proposed model is (almost) independent of the current model, so that ``distant'' moves in the model space are encouraged. This can be advantageous in comparison to Gibbs samplers and Metropolis-Hastings algorithms based on local proposal distributions, which may yield larger acceptance rates but are more prone to be stuck in local modes of the posterior model distribution. In future work one may also consider combinations of the adaptive independent proposals in MAdaSub with adaptive local proposals as for example in \citet{lamnisos2013} and \citet{zanella2019}. While MAdaSub yields competitive results without the use of Rao-Blackwellization compared to the related adaptive algorithms of \cite{griffin2018}, the incorporation of Rao-Blackwellized estimates of posterior inclusion probabilities in the burn-in phase or as initial proposal probabilities may further increase the speed of convergence of MAdaSub. Finally, the extension of MAdaSub to settings with non-conjugate priors is interesting to be investigated, for example by considering data augmentation approaches with additional latent variables or by incorporating reversible-jump moves \citep{green1995, wan2021}.

\bibliographystyle{Chicago}
{\footnotesize
\setlength{\bibsep}{0.5pt}
\bibliography{AdaptiveMCMC}
}







\counterwithout{theorem}{section}
\counterwithout{notation}{section}

\newtheorem{notationS}{Notation}

\theoremstyle{plain}
\newtheorem{theoremS}{Theorem}
\newtheorem{lemmaS}{Lemma}
\newtheorem{theoremrep}{Theorem}
\newtheorem{corollaryrep}{Corollary}

\renewcommand{\thetheoremS}{A.\arabic{theoremS}}
\renewcommand{\thelemmaS}{A.\arabic{lemmaS}}
\renewcommand{\thenotationS}{A.\arabic{notationS}}
\renewcommand*{\thesection}{A\arabic{section}}
\renewcommand*{\thefigure}{A.\arabic{figure}}

\clearpage

\appendix


\section{Ergodicity of the MAdaSub algorithm}\label{sec:ergodicity}

In this section we present a detailed proof for the ergodicity of the serial MAdaSub algorithm (see Theorem~\ref{thmA:MAdaSub}), i.e.\ we show that ``in the limit'' MAdaSub samples from the targeted posterior model distribution \(\pi(\cdot\,|\,\mathcal{D})\) despite the continuing adaptation of the algorithm. We will make use of a general ergodicity result for adaptive MCMC algorithms by \citet{roberts2007}. In order to state the result directly for the specific setting of the MAdaSub algorithm, we first introduce some notation. 

\begin{notationS}\label{not:MAdaSubSupp}
\begin{enumerate}
\item[(a)] In the following, the models \(S^{(0)},S^{(1)},S^{(2)},\dots\) generated by the MAdaSub algorithm (see Algorithm~1 of the main document) should be viewed as random variables with values in the model space \(\mathcal{M}=\{S;\, S\subseteq\{1,\dots,p\}\}\). Furthermore, the (truncated) vectors of proposal probabilities \(\tilde{\bs r}^{(t)}=\left(\tilde{\bs r}_1^{(t)},\dots,\tilde{\bs r}_p^{(t)}\right)^T\), \(t\in\N\) should be viewed as 
random vectors 
with values in the compact set \(\mathcal{I}^p=[\epsilon,1-\epsilon]^p\). 


\item[(b)] For a (current) model \(S\in\mathcal{M}\) and a vector of proposal probabilities \(\tilde{\bs{r}}\in[\epsilon,1-\epsilon]^p\), let \(P(\cdot\,|\,S;\tilde{\bs r})\) denote the one-step transition kernel of MAdaSub, i.e.\ for iteration \(t\in\N\) of MAdaSub and a subset of models \(A'\subseteq\mathcal{M}\) we have
\begin{equation}
P(A'\,|\,S;\tilde{\bs r}) = P\left(S^{(t)}\in A'\,\big|\, S^{(t-1)}=S, \tilde{\bs r} ^{(t-1)} = \tilde{\bs r}\right) \,.
\label{eq:transition} 
\end{equation}
In particular, for \(S'\in\mathcal{M}\), let
\(
P(S'\,|\,S;\tilde{\bs r}) \equiv  P(\{S'\}\,|\,S;\tilde{\bs r})
\)
denote the probability that the next state of the MAdaSub chain is \(S^{(t)}=S'\), given the current model \(S^{(t-1)}=S\) and the current vector of proposal probabilities \(\tilde{\bs r} ^{(t-1)}=\tilde{\bs r}\). Note that for \(\tilde{\bs{r}}\in[\epsilon,1-\epsilon]^p\) and \(S,S'\in\mathcal{M}\) with \(S\neq S'\) we have 
\begin{equation}
P(S'\,|\,S;\tilde{\bs r}) = q(S'; \tilde{\bs r}) \, \alpha(S'\,|\,S;\tilde{\bs r})\, ,
\end{equation}
where \(q(S'; \tilde{\bs r})\) is the probability of proposing the model \(S'\) and \(\alpha(S'\,|\,S;\tilde{\bs r})\) is the corresponding acceptance probability.

\item[(c)] 
For \(t\in\N\), \(S\in\mathcal{M}\), \(A'\subseteq\mathcal{M}\) and \(\tilde{\bs{r}}\in[\epsilon,1-\epsilon]^p\) let 
\begin{equation}
P^{(t)}(A'\,|\,S;\tilde{\bs r}) := P\left(S^{(t)}\in A'\,\big|\, S^{(0)}=S, \tilde{\bs r} ^{(0)} = \ldots = \tilde{\bs r} ^{(t-1)} = \tilde{\bs r}\right)
\end{equation}
denote the \(t\)-step transition kernel of MAdaSub when the vector of proposal probabilities \(\tilde{\bs{r}}\) is fixed (i.e.\ not adapted during the algorithm). Similarly, let
\begin{equation}
Q^{(t)}(A'\,|\,S;\tilde{\bs r}) := P\left(S^{(t)}\in A'\,\big|\, S^{(0)}=S, \tilde{\bs r} ^{(0)} = \tilde{\bs r}\right)
\label{eq:transitiont} 
\end{equation}
denote the \(t\)-step transition kernel for the first \(t\) iterations of MAdaSub, given only the initial conditions \(S^{(0)}=S\) and \(\tilde{\bs r} ^{(0)} = \tilde{\bs r}\). 
\end{enumerate}
\end{notationS}

The following theorem provides the ergodicity result of \citet[Theorem~1]{roberts2007} 
adjusted to the specific setting of MAdaSub.  

\begin{theoremS}[\citealp{roberts2007}]\label{thm:roberts}
Consider the MAdaSub algorithm with initial parameters \(\bs r^{(0)}\in(0,1)^p\), \(L_j>0\) and \(\epsilon\in(0,0.5)\). 
Suppose that for each fixed vector of proposal probabilities \(\tilde{\bs r}\in[\epsilon,1-\epsilon]^p\), the one-step kernel \(P(\cdot\,|\,\cdot\,;\tilde{\bs r})\) of MAdaSub is stationary for the target distribution \(\pi(\cdot\,|\,\mathcal{D})\), i.e.\ for all \(S'\in\mathcal{M}\) we have 
\begin{equation}
\pi(S'\,|\,\mathcal{D}) = \sum_{S\in\mathcal{M}} P(S'\,|\,S;\tilde{\bs r}) \, \pi(S\,|\,\mathcal{D}) \,.
\end{equation}
Further suppose that the following two conditions hold: 
\begin{enumerate}
\item[(a)] The \textbf{simultaneous uniform ergodicity} condition is satisfied, i.e.\ for all \(\delta>0\), there exists an integer \(T\in\N\) such that 
\begin{equation}\left\lVert P^{(T)}(\cdot\,|\,S; \tilde{\bs r}) -\pi(\cdot\,|\,\mathcal{D})\right\rVert_{TV} \leq \delta\end{equation} 
for all \(S\in\mathcal{M}\) and \(\tilde{\bs r}\in[\epsilon,1-\epsilon]^p\), where \(\lVert P_1-P_2\rVert_{TV} = \sup_{A\in\mathfrak{A}}|P_1(A)-P_2(A)|\) denotes the total variation distance between two distributions \(P_1\) and \(P_2\) defined on some common measurable space \((\Omega,\mathfrak{A})\).
\item[(b)] The \textbf{diminishing adaptation} condition is satisfied, i.e.\ we have 
\begin{equation} \max_{S \in\mathcal{M}} ~\left\lVert\, P\big(\cdot\big|\, S; \tilde{ \bs r}^{(t)}\big) - P\big(\cdot\big|\, S; \tilde{ \bs r}^{(t-1)}\big)\,\right\rVert_{TV} \overset{\text{P}}{\rightarrow} 0 \,,  ~~ t\rightarrow\infty \,,
\end{equation} 
where \(\tilde{ \bs r}^{(t)}\) and \(\tilde{ \bs r}^{(t-1)}\) are random vectors of proposal probabilities induced by the MAdaSub algorithm (see Notation~\ref{not:MAdaSubSupp}).  
\end{enumerate} 
Then the MAdaSub algorithm is \textbf{ergodic}, i.e.\ for all \(S\in\mathcal{M}\) and \(\tilde{\bs r}\in[\epsilon,1-\epsilon]^p\) we have
\begin{equation}
\left\lVert Q^{(t)}(\cdot\,|\, S; \tilde{\bs r}) - \pi(\cdot\,|\,\mathcal{D}) \right\rVert_{TV} \rightarrow 0, ~~ t\rightarrow\infty \, . 
\end{equation}
Furthermore, the \textbf{weak law of large numbers} holds for MAdaSub, i.e.\ for any function \(g:\mathcal{M}\rightarrow\R\) we have 
\begin{equation}
 \frac{1}{t}\sum_{i=1}^t g(S^{(i)}) \overset{\text{P}}{\rightarrow} E[g\,|\,\mathcal{D}] \,,~~ t\rightarrow\infty \,,
\end{equation}
where \(E[g\,|\,\mathcal{D}]=\sum_S g(S)\pi(S\,|\,\mathcal{D})\) denotes the posterior expectation of \(g\).
\end{theoremS}

In the following we will show that MAdaSub satisfies both the simultaneous uniform ergodicity condition and the diminishing adaptation condition, so that Theorem~\ref{thm:roberts} can be applied.   

\begin{lemmaS}\label{lemma:simuniform}
The simultaneous uniform ergodicity condition is satisfied for the MAdaSub algorithm for all choices of \(\bs r^{(0)}\in(0,1)^p\), \(L_j>0\) and \(\epsilon\in(0,0.5)\).
\end{lemmaS}

\begin{proof}
Here we make use of a very similar argumentation as in the proof of Lemma~1 in \citet{griffin2018}. 
We show that \(\mathcal{M}\) is a \textit{1-small set} (see \citealp[Section~3.3]{roberts2004})
, i.e.\ there exists \(\beta>0\) and a probability measure \(\nu\) on \(\mathcal{M}\) such that \(P(A'\,|\,S; \tilde{\bs r})\geq\beta\nu(A')\) for all \(S\in\mathcal{M}\), \(A'\subseteq\mathcal{M}\) and \(\tilde{\bs r}\in [\epsilon,1-\epsilon]^p\). Then by Theorem~8 in \citet{roberts2004}, the simultaneous uniform ergodicity condition is satisfied. \\
In order to prove that \(\mathcal{M}\) is 1-small (note that \(\mathcal{M}\) is finite), it suffices to show that there exists a constant \(\beta_0>0\) such that \(P(S'\,|\,S; \tilde{\bs r})\geq\beta_0\) for all \(S,S'\in\mathcal{M}\) and all \(\tilde{\bs r}\in [\epsilon,1-\epsilon]^p\). Indeed, for \(S,S'\in\mathcal{M}\) and \(\tilde{\bs r}\in [\epsilon,1-\epsilon]^p\) it holds
\vspace{-2mm}
\begin{align*}
P(S'\,|\,S; \tilde{\bs r}) &\geq q(S'; \tilde{\bs r}) \, \alpha(S'\,|\,S; \tilde{\bs r}) \\
                           &= \Bigg(\prod_{j\in S'} \underbrace{\tilde{r}_j}_{\geq\epsilon} \Bigg) \Bigg(\prod_{j\in\mathcal{P}\setminus S'} \underbrace{(1-\tilde{r}_j)}_{\geq\epsilon} \Bigg) \min\left\{\frac{\pi(S'\,|\,\mathcal{D}) \, q(S; \tilde{\bs r})}  {\pi(S\,|\,\mathcal{D}) \, q(S'; \tilde{\bs r} )} ,1 \right\} \\
													&\geq \epsilon^p \pi(S'\,|\,\mathcal{D}) \, q(S; \tilde{\bs r}) \geq \epsilon^{2p} \min_{S\in\mathcal{M}} \pi(S\,|\,\mathcal{D}) =: \beta_0 \,.
\end{align*}
This completes the proof.
\end{proof}

In order to show that the diminishing adaptation condition is satisfied for the MAdaSub algorithm, we will make repeated use of the following simple observation. 

\begin{lemmaS}\label{lemma:help}
Let \(m\in\N\) be fixed. For \(j\in\{1,\dots,m\}\) let \(\left(a_j^{(t)}\right)_{t\in\N_0}\) be bounded sequences of real numbers \(a_j^{(t)}\in\R\) with \(|a_j^{(t)}-a_j^{(t-1)}|\rightarrow 0\) for \(t\rightarrow\infty\). Then we have
\begin{equation} \left|\prod_{j=1}^m a_j^{(t)} - \prod_{j=1}^m a_j^{(t-1)} \right| \rightarrow 0,~~ t\rightarrow\infty \,. \label{eq:prod0}\end{equation}
\end{lemmaS}

\begin{proof}
Since \(\left(a_j^{(t)}\right)_{t\in\N_0}\) are bounded sequences, there are constants \(L_j>0\) so that \(|a_j^{(t)}|\leq L_j\) for all \(t\in\N_0\) and \(j\in\{1,\dots,m\}\). 
We proceed by induction on \(m\in\N\): equation~(\ref{eq:prod0}) obviously holds for \(m=1\). Now suppose that the assertion holds for \(m-1\) and we want to show that it also holds for \(m\). Then we have 
\begin{align*}
\left| \prod_{j=1}^m a_j^{(t)} - \prod_{j=1}^m a_j^{(t-1)} \right|  
\leq & \left| a_m^{(t)} \prod_{j=1}^{m-1} a_j^{(t)} - a_m^{(t-1)} \prod_{j=1}^{m-1} a_j^{(t)}\right| + \left| a_m^{(t-1)} \prod_{j=1}^{m-1} a_j^{(t)}  - a_m^{(t-1)} \prod_{j=1}^{m-1} a_j^{(t-1)}  \right| \\
= & \underbrace{\prod_{j=1}^{m-1} \left| a_j^{(t)} \right|}_{\leq \prod_{j=1}^{m-1} L_j} \times \underbrace{\left|a_m^{(t)}- a_m^{(t-1)} \right|}_{\rightarrow 0} + \underbrace{\left| a_m^{(t-1)}\right|}_{\leq L_m} \times \underbrace{\left|\prod_{j=1}^{m-1} a_j^{(t)} - \prod_{j=1}^{m-1} a_j^{(t-1)} \right|}_{\rightarrow 0} \overset{t\rightarrow\infty}{\rightarrow} 0 \, . 
\end{align*}
\end{proof}

\begin{lemmaS}\label{lemma:dimada}
Consider the application of the MAdaSub algorithm on a given dataset \(\mathcal{D}\) with some tuning parameter choices \(\bs r^{(0)}\in(0,1)^p\), \(L_j>0\) and \(\epsilon\in(0,0.5)\). 
Then, for \(j\in\mathcal{P}\), we have
\begin{equation} \left| \tilde{r}_j^{(t)} - \tilde{r}_j^{(t-1)} \right| \overset{\text{a.s.}}{\rightarrow} 0, ~~t\rightarrow\infty \, . \end{equation} 
Furthermore, for all \(S,S'\in\mathcal{M}\) it holds 
\begin{equation} \left| P(S'\,|\,S; \tilde{\bs r}^{(t)}) - P(S'\,|\,S; \tilde{\bs r}^{(t-1)}) \right| \overset{\text{a.s.}}{\rightarrow} 0  , ~~t\rightarrow\infty \,. \label{eq:dimada2}\end{equation} 
In particular, MAdaSub fulfils the diminishing adaptation condition. 
\end{lemmaS}

\begin{proof}
For \(j\in\mathcal{P}\) we have 
\vspace{-1mm}
\begin{align*}
\Big|\tilde{r}_j^{(t)} - \tilde{r}_j^{(t-1)}\Big| 
&\leq \Big|r_j^{(t)} - r_j^{(t-1)}\Big| \\
 &\leq \left| \frac{L_jr_j^{(0)} + \sum_{i=1}^t \mathbbm{1}_{S^{(i)}}(j)}{L_j + t} - \frac{L_jr_j^{(0)} + \sum_{i=1}^{t-1} \mathbbm{1}_{S^{(i)}}(j)}{L_j+t-1}  \right| \\
&\leq \left| \frac{L_jr_j^{(0)} +  \sum_{i=1}^t \mathbbm{1}_{S^{(i)}}(j)}{L_j + t} - \frac{L_jr_j^{(0)} + \sum_{i=1}^{t} \mathbbm{1}_{S^{(i)}}(j)}{L_j+t-1}  \right| \\
&~~~ +  \left| \frac{L_jr_j^{(0)} + \sum_{i=1}^{t} \mathbbm{1}_{S^{(i)}}(j)}{L_j+t-1} - \frac{L_jr_j^{(0)} +\sum_{i=1}^{t-1} \mathbbm{1}_{S^{(i)}}(j)}{L_j+t-1}  \right| \\
&\leq \underbrace{\frac{L_jr_j^{(0)} + \sum_{i=1}^t \mathbbm{1}_{S^{(i)}}(j)}{L_j+t}}_{\in(0,1)} \times \underbrace{\frac{1}{L_j+t-1}}_{\rightarrow 0} + \underbrace{\frac{1}{L_j+t-1}}_{\rightarrow 0} \overset{\text{a.s.}}{\rightarrow} 0 , ~~t\rightarrow\infty\,. 
\end{align*}
With Lemma~\ref{lemma:help} (set \(m=p\) and note that the number of variables \(p=|\mathcal{P}|\) is fixed for the given dataset) we conclude that for \(V\in\mathcal{M}\) it holds 
\begin{align} 
 \left| q(V; \tilde{\bs r}^{(t)}) - q(V; \tilde{\bs r}^{(t-1)}) \right| = \left| \prod_{j\in V} \tilde{r}_j^{(t)} \prod_{j\in\mathcal{P}\setminus V} \left(1-\tilde{r}_j^{(t)}\right) - \prod_{j\in V} \tilde{r}_j^{(t-1)} \prod_{j\in\mathcal{P}\setminus V} \left(1-\tilde{r}_j^{(t-1)}\right) \right| \overset{\text{a.s.}}{\rightarrow}  0\,. \label{eq:zero1}
\end{align} 
Let \(S,S'\in\mathcal{M}\) and suppose that \(S\neq S'\). Then we have 
\begin{align}
\left| P(S'\,|\,S; \tilde{\bs r}^{(t)}) - P(S'\,|\,S; \tilde{\bs r}^{(t-1)}) \right| = \left| q(S'; \tilde{\bs r}^{(t)})\alpha(S'\,|\,S; \tilde{\bs r}^{(t)}) - q(S'; \tilde{\bs r}^{(t-1)})\alpha(S'\,|\,S; \tilde{\bs r}^{(t-1)}) \right| \,. \label{eq:zero2}
\end{align}
Note that \(q(S'; \tilde{\bs r}^{(t)})\in[\epsilon^p,(1-\epsilon)^p]\) and \( \alpha(S'\,|\,S; \tilde{\bs r}^{(t)}) \in [0,1]\) for all \(t\in\N_0\). Furthermore, we have already shown that \( \left| q(S'; \bs r^{(t)}) - q(S'; \bs r^{(t-1)}) \right| \overset{\text{a.s.}}{\rightarrow}  0\) for all \(S'\in\mathcal{M}\). Therefore, we also have 
\begin{align}
  \left| \alpha(S'\,|\,S; \tilde{\bs r}^{(t)}) - \alpha(S'\,|\,S; \tilde{\bs r}^{(t-1)}) \right| 
&\leq \left|  \frac{C(S') \, q(S; \tilde{\bs r}^{(t)})}  {C(S) \, q(S'; \tilde{\bs r}^{(t)} )} - \frac{C(S') \, q(S; \tilde{\bs r}^{(t-1)})}  {C(S) \, q(S'; \tilde{\bs r}^{(t-1)} )}   \right| \nonumber \\
&= \frac{C(S')}{C(S)}  \left|  \frac{ q(S; \tilde{\bs r}^{(t)})}  {q(S'; \tilde{\bs r}^{(t)} )} - \frac{q(S; \tilde{\bs r}^{(t-1)})}  { q(S'; \tilde{\bs r}^{(t-1)} )}  \right| \overset{\text{a.s.}}{\rightarrow}  0 \,, \label{eq:zero3}
\end{align}
where we made use of Lemma~\ref{lemma:help} with \(m=2\) and  \[ a_1^{(t)} = q(S; \tilde{\bs r}^{(t)})\in[\epsilon^p,(1-\epsilon)^p] ~~\text{  and  } ~~ a_2^{(t)} = \frac{1}{q(S'; \tilde{\bs r}^{(t)})}\in[(1-\epsilon)^{-p},\epsilon^{-p}] \, , ~~ t\in\N_0 \, ,\]
noting that 
\[ \left|a_2^{(t)}-a_2^{(t-1)} \right| = \underbrace{\frac{1}{q(S'; \tilde{\bs r}^{(t)})q(S'; \tilde{\bs r}^{(t-1)})}}_{\leq\epsilon^{-2p}} \left|q(S'; \tilde{\bs r}^{(t)})- q(S'; \tilde{\bs r}^{(t-1)}) \right| \overset{\text{a.s.}}{\rightarrow}  0 \,.\]
Again by using Lemma~\ref{lemma:help} and combining equations~(\ref{eq:zero1}), (\ref{eq:zero2}) and (\ref{eq:zero3}) we conclude that 
\[ \left| P(S'\,|\,S; \tilde{\bs r}^{(t)}) - P(S'\,|\,S; \tilde{\bs r}^{(t-1)}) \right| \overset{\text{a.s.}}{\rightarrow}  0  \,.\]
Finally, we consider the case \(S=S'\). Then it holds
\begin{align*}
\left| P(S\,|\,S; \tilde{\bs r}^{(t)}) - P(S\,|\,S; \tilde{\bs r}^{(t-1)}) \right| &= \Big| 1- \sum_{S'\neq S} P(S'\,|\,S; \tilde{\bs r}^{(t)}) - \Big( 1- \sum_{S'\neq S}  P(S'\,|\,S; \tilde{\bs r}^{(t-1)}) \Big) \Big| \\
&\leq  \sum_{S'\neq S} \Big| P(S'\,|\,S; \tilde{\bs r}^{(t)}) - P(S'\,|\,S; \tilde{\bs r}^{(t-1)}) \Big| \overset{\text{a.s.}}{\rightarrow}  0 \,.
\end{align*}
Thus we have shown that equation~(\ref{eq:dimada2}) holds for all \(S,S'\in\mathcal{M}\). In particular, we conclude that the diminishing adaptation condition is satisfied for MAdaSub (recall that almost sure convergence implies convergence in probability). 
\end{proof}

\begin{theorem}\label{thmA:MAdaSub}
The MAdaSub algorithm (Algorithm~1) is ergodic for all choices of \(\bs r^{(0)}\in(0,1)^p\), \(L_j>0\) and \(\epsilon\in(0,0.5)\) and fulfils the weak law of large numbers. 
\end{theorem}

\begin{proof}
The MAdaSub algorithm fulfils the simultaneous uniform ergodicity condition (see Lemma~\ref{lemma:simuniform}) and the diminishing adaptation condition (see Lemma~\ref{lemma:dimada}). 
Furthermore, for each fixed \(\tilde{\bs r}\in[\epsilon,1-\epsilon]^p\), the corresponding transition kernel  \(P(\cdot\,|\,\cdot\,;\tilde{\bs r})\) is induced by a simple Metropolis-Hastings step and therefore has the desired target distribution \(\pi(\cdot\,|\,\mathcal{D})\) as its stationary distribution. Hence, by Theorem~\ref{thm:roberts} the MAdaSub algorithm is ergodic and fulfils the weak law of large numbers.
\end{proof}

\begin{corollary}
For all choices of \(\bs r^{(0)}\in(0,1)^p\), \(L_j>0\) and \(\epsilon\in(0,0.5)\), the proposal probabilities \(r_j^{(t)}\) of the explanatory variables \(X_j\) in MAdaSub converge (in probability) to the respective posterior inclusion probabilities \(\pi_j=\pi(j\in S\,|\,\mathcal{D})\), i.e. for all \(j\in\mathcal{P}\) it holds that 
\( r_j^{(t)} \overset{\text{P}}{\rightarrow} \pi_j \) as~\(t\rightarrow\infty\).
\end{corollary}

\begin{proof}
Since MAdaSub fulfils the weak law of large numbers (Theorem~\ref{thmA:MAdaSub}), for \(j\in\mathcal{P}\) it holds that 
\[ \frac{1}{t}\sum_{i=1}^t \mathbbm{1}_{S^{(i)}}(j) \overset{\text{P}}{\rightarrow} \pi_j ,~~ t\rightarrow\infty \,.\]
Hence, for \(j\in\mathcal{P}\), we also have
\[r_j^{(t)} = \frac{L_j r_j^{(0)}+\sum_{i=1}^t \mathbbm{1}_{S^{(i)}}(j)}{L_j+t} \overset{\text{P}}{\rightarrow} \pi_j ,~~ t\rightarrow\infty \,.\]
\end{proof}

\clearpage

\setcounter{algorithm}{1}
\begin{samepage}
\section{Algorithmic details of parallel version of MAdaSub} \label{sec:detailsparallel}

\begin{algorithm}[H]
\caption{Parallel version of MAdaSub}\label{algo:parallel} 
\vspace{-0.1cm}
\begin{flushleft} \textbf{Input:} \end{flushleft}
\vspace{-7mm}
\begin{itemize}
\item Number of workers \(K\in\N\).
\vspace{-3mm}
\item Number of rounds \(R\in\N\).
\vspace{-3mm}
\item Number of iterations per round \(T\in\N\). 
\vspace{-3mm}
\item Data $\mathcal{D}=(\bs X,\bs y)$.
\vspace{-3mm}
\item (Approximate) kernel of posterior \(\pi(S\,|\,\mathcal{D})\propto \pi( \bs y \,|\,\bs X, S) \, \pi( S ) \) for \(S\in\mathcal{M}\).
\vspace{-3mm}
\item Vector of initial proposal probabilities \(\bs r^{(k,0)} = \left(r_1^{(k,0)},\dots,r_p^{(k,0)}\right)^T\in(0,1)^p\) for each worker \(k=1,\dots,K\). 
\vspace{-3mm}
\item Adaptation parameters $L^{(k)}_j>0$  for \(j\in\mathcal{P}\) and each worker \(k=1,\dots,K\). 
\vspace{-3mm}
\item Constant \(\epsilon\in(0,0.5)\) (chosen to be small, e.g.\ \(\epsilon\leq\frac{1}{p}\)).
\vspace{-3mm}
\item Starting points \(S^{(k,0)}\in\mathcal{M}\) for \(k=1,\dots,K\) (optional).
\end{itemize}

\vspace{-4mm}
\begin{flushleft} \textbf{Algorithm: } \end{flushleft}
\vspace{-6mm}
\begin{enumerate}
\item[(1)] 
Set \(\bar{\bs r}^{(k,0)} = \bs r^{(k,0)}\) for \(k=1,\dots,K\).  \\[1mm]
For \(k=1,\dots,K\): If starting point \(S^{(k,0)}\) not specified:
\par
\begingroup
\leftskip=0.5cm 
\noindent Sample $b_j^{(k,0)}\sim\text{Bernoulli}\left(r_j^{(k,0)}\right)$ independently for $j\in\mathcal{P}$. \\
Set \(S^{(k,0)} = \{j\in\mathcal{P};~b_j^{(k,0)}=1\}\).
\par
\endgroup 
\item[(2)] For $m=1,\dots,R$: (for each round)
\vspace{-1mm}
\begin{enumerate}
\item[(a)] For \(k=1,\dots,K\): (for each worker in parallel)
  \par
\begingroup
\leftskip=0.4cm 
\vspace{1mm}
\noindent 	Run MAdaSub (Algorithm~1) on worker \(k\) for \(T\) iterations with 
         \begin{itemize}
				 \item starting point \(S^{(k,(m-1)T)}\),
				 \item initial proposal probabilities \(\bar{\bs r}^{(k,m-1)}\),
				 \item initial adaptation parameters \(L^{(k)}_j + (m-1)TK\), for \(j\in\mathcal{P}\).
	       \end{itemize}
					Output: Sampled models \(S^{(k,(m-1)T + t)}\) for \(t=1,\dots,T\).
\par
\endgroup
\vspace{1mm}
\item[(b)] Exchange information between workers: \\
For \(k=1,\dots,K\) compute \(\bar{\bs r}^{(k,m)} = \left(\bar{r}_1^{(k,m)},\dots,\bar{r}_p^{(k,m)}\right)^T\) with 
\[\bar{ r}_j^{(k,m)} = \frac{ L_j^{(k)} r_j^{(k,0)} + \sum_{t=1}^{mT} \sum_{l=1}^K \mathbbm{1}_{S^{(l,t)}}(j) }{ L_j^{(k)} + mTK },~~ j\in\mathcal{P} \,.\]
\end{enumerate}
\end{enumerate}

\vspace{-5mm}
\begin{flushleft} \textbf{Output:}  \end{flushleft}
\vspace{-7mm}
\begin{itemize}
\item For each worker \(k=1,\dots,K\) approximate sample \(S^{(k,b+1)},\dots,S^{(k,RT)}\) from posterior distribution \(\pi(\cdot\,|\,\mathcal{D})\), after burn-in period of length \(b\).
\end{itemize}
\end{algorithm}

\end{samepage}

\clearpage

\section{Ergodicity of parallel version of MAdaSub}\label{sec:ergodicityparallel}

In this section we extend the ergodicity result for the serial MAdaSub algorithm (Algorithm~1) of Section~\ref{sec:ergodicity} to the parallel version of MAdaSub (Algorithm~\ref{algo:parallel}).

\begin{theorem}\label{thm:parallelSup}
Consider the parallel version of MAdaSub (Algorithm~\ref{algo:parallel}). Then, for each worker \(k\in\{1,\dots,K\}\) and all choices of \(\bs r^{(k,0)}\in(0,1)^p\), \(L^{(k)}_j>0\), \(j\in\mathcal{P}\) and \(\epsilon\in(0,0.5)\), each induced chain \(S^{(k,0)},S^{(k,1)},\dots\) of the workers \(k=1,\dots,K\) is ergodic and fulfils the weak law of large numbers.
\end{theorem}

\begin{proof} 
The proof of the simultaneous uniform ergodicity condition for each of the parallel chains is along the lines of the proof for the serial version of MAdaSub (see Lemma~\ref{lemma:simuniform}). As before, we can conclude with Theorem~\ref{thm:roberts} that each parallel chain is ergodic and fulfils the weak law of large numbers, provided that the diminishing adaptation condition is also satisfied for each of the parallel chains. \\
In order to show the diminishing adaptation condition for the chain on worker \(k\in\{1,\dots,K\}\) it suffices to show that for \(j\in\mathcal{P}\) it holds 
\begin{equation} \left| r_j^{(k,t)} - r_j^{(k,t-1)} \right| \overset{\text{a.s.}}{\rightarrow} 0, ~~t\rightarrow\infty \, , \label{eq:dimparallel} \end{equation} 
where 
\begin{equation} r_j^{(k,t)} = \frac{L_j^{(k)} r_j^{(k,0)} + \sum_{l=1,l\neq k}^K \sum_{i=1}^{ \lfloor \frac{t}{T} \rfloor T} \mathbbm{1}_{S^{(l,i)}}(j) + \sum_{i=1}^t \mathbbm{1}_{S^{(k,i)}}(j) } {L_j^{(k)} + \lfloor \frac{t}{T} \rfloor T (K-1) + t } \label{eq:rparallel}\end{equation}
denotes the proposal probability of variable \(X_j\) after \(t\) iterations of the chain on worker \(k\); the remaining steps of the proof are analogous to the proof of diminishing adaptation for the serial version of MAdaSub (see Lemma~\ref{lemma:dimada}). Note that in equation~(\ref{eq:rparallel}) we make use of the convention that \(\sum_{i=a}^b c_i = 0 \) for \(b<a\); additionally, \(\lfloor c \rfloor\in\N\) denotes the greatest integer less than or equal to \(c\in\R\). Furthermore, note that for \(t=mT\) with \(m\in\N\) it holds \(r_j^{(k,t)} = \bar{r}_j^{(k,m)}\) for \(j\in\mathcal{P}, k\in\{1,\dots,K\}\). \\
Using the triangle inequality (compare proof of Lemma~\ref{lemma:dimada}) and noting that for all \(t,T\in\N\) we have \(\lfloor \frac{t}{T} \rfloor - \lfloor \frac{t-1}{T} \rfloor \leq 1\), we conclude that for \(k\in\{1,\dots,K\}\) it holds
\begin{align*}
& ~~~ \left| r_j^{(k,t)} - r_j^{(k,t-1)} \right| \\[2mm]
&= \Bigg| \frac{L_j^{(k)} r_j^{(k,0)} + \sum_{l=1,l\neq k}^K \sum_{i=1}^{ \lfloor \frac{t}{T} \rfloor T} \mathbbm{1}_{S^{(l,i)}}(j) + \sum_{i=1}^t \mathbbm{1}_{S^{(k,i)}}(j) } {L_j^{(k)} + \lfloor \frac{t}{T} \rfloor T (K-1) + t } \\[1mm]
&~~~~ - \frac{L_j^{(k)} r_j^{(k,0)} + \sum_{l=1,l\neq k}^K \sum_{i=1}^{ \lfloor \frac{t-1}{T} \rfloor T} \mathbbm{1}_{S^{(l,i)}}(j) + \sum_{i=1}^{t-1} \mathbbm{1}_{S^{(k,i)}}(j) } {L_j^{(k)} + \lfloor \frac{t-1}{T} \rfloor T (K-1) + t-1 } \Bigg| \\[3mm]
&\leq \underbrace{\frac{L_j^{(k)} r_j^{(k,0)} + \sum_{l=1,l\neq k}^K \sum_{i=1}^{ \lfloor \frac{t}{T} \rfloor T } \mathbbm{1}_{S^{(l,i)}}(j) + \sum_{i=1}^t \mathbbm{1}_{S^{(k,i)}}(j) } {L_j^{(k)} + \lfloor \frac{t}{T} \rfloor T (K-1) + t } }_{\in(0,1)} 
\times \underbrace{\frac{(K-1)T\left(\lfloor \frac{t}{T} \rfloor - \lfloor \frac{t-1}{T} \rfloor \right) + 1 } { L_j^{(k)} + \lfloor \frac{t-1}{T} \rfloor T (K-1) + t-1 }}_{\rightarrow 0} \\[1mm]
&~~~~+ \underbrace{\frac{(K-1)T +1}{L_j^{(k)} + \lfloor \frac{t-1}{T} \rfloor T (K-1) + t-1 }}_{\rightarrow 0} \overset{\text{a.s.}}{\rightarrow} 0 , ~~t\rightarrow\infty\,.
\end{align*} 
Thus, we have shown that equation~(\ref{eq:dimparallel}) holds and this completes the proof.
\end{proof}

\begin{corollary}
Consider the parallel version of MAdaSub (Algorithm~\ref{algo:parallel}). 
Then, for each worker \(k\in\{1,\dots,K\}\) and all choices of \(\bs r^{(k,0)}\in(0,1)^p\), \(L^{(k)}_j>0\), \(j\in\mathcal{P}\) and \(\epsilon\in(0,0.5)\), the proposal probabilities \(\bar{r}_j^{(k,m)}\) of the explanatory variables \(X_j\) converge (in probability) to the respective posterior inclusion probabilities \(\pi_j=\pi(j\in S\,|\,\mathcal{D})\), i.e. for all \(j\in\mathcal{P}\) and \(k=1,\dots,K\) it holds that 
\( \bar{r}_j^{(k,m)} \overset{\text{P}}{\rightarrow} \pi_j \) as \(m\rightarrow\infty\).
\end{corollary}

\begin{proof}
Since each chain in the parallel MAdaSub algorithm fulfils the weak law of large numbers (Theorem~\ref{thm:parallelSup}), for \(j\in\mathcal{P}\) and \(k\in\{1,\dots,K\}\) it holds that 
\[ \frac{1}{mT}\sum_{i=1}^{mT} \mathbbm{1}_{S^{(k,i)}}(j) \overset{\text{P}}{\rightarrow} \pi_j ,~~ m\rightarrow\infty \,.\]
Hence, for \(j\in\mathcal{P}\) and \(k\in\{1,\dots,K\}\), we also have
\begin{align*}
\bar{r}_j^{(k,m)}  &= \frac{ L_j^{(k)} r_j^{(k,0)} + \sum_{t=1}^{mT} \sum_{l=1}^K \mathbbm{1}_{S^{(l,t)}}(j) }{ L_j^{(k)} + mTK } \\ 
 &= \frac{ \frac{L_j^{(k)} r_j^{(k,0)}}{mTK} + \frac{1}{K} \sum_{l=1}^K \frac{1}{mT} \sum_{t=1}^{mT} \mathbbm{1}_{S^{(l,t)}}(j) } { \frac{L_j^{(k)}}{mTK} + 1 }  \overset{\text{P}}{\rightarrow} \pi_j ,~~ m\rightarrow\infty \,.
\end{align*}\end{proof}

\clearpage

\section{Further approaches related to MAdaSub} \label{sec:comparisonsSup}

\citet{clyde2011} propose a Bayesian Adaptive Sampling (BAS) algorithm which is based on sampling without replacement from the posterior model distribution, where the individual sampling probabilities of the variables are adapted during the algorithm in such a manner that they converge against the posterior inclusion probabilities. By construction, if the number of iterations is equal to the number of possible models, the BAS algorithm enumerates all possible models. However, since BAS samples without replacement, it has to be ensured that no model is sampled twice and therefore, after each iteration of the algorithm, the sampling probabilities of some of the remaining models have to be renormalized. 
Additionally, BAS differs from the other methods discussed in Section~3.2 since it is not an MCMC algorithm and may yield biased estimates of posterior inclusion probabilities after a limited number of iterations. 

Another related adaptive method for Bayesian variable selection has been proposed by \citet{ji2013}. They consider an adaptive independence Metropolis-Hastings algorithm for sampling directly from the posterior distribution of the regression coefficients \(\bs\beta=(\beta_1,\dots,\beta_p)^T\), assuming that the prior of \(\beta_j\) for \(j\in\mathcal{P}\) is given by a mixture of a point-mass at zero (indicating that the corresponding variable~\(X_j\) is not included in the model) and a continuous normal distribution (indicating that variable~\(X_j\) is ``relevant''). Mixtures of normal distributions are used as proposals in the Metropolis-Hastings step, which are adapted during the algorithm to minimize the Kullback-Leibler divergence from the target distribution. The considered family of mixture distributions should ideally have a sufficient number of mixture components to be able to approximate the 
multimodal posterior distribution of~\(\bs\beta\). In comparison, MAdaSub focuses on sampling from the discrete model distribution and makes use of independent Bernoulli distributions as approximations to the targeted posterior model distribution, while the updating scheme is motivated in a Bayesian way. Further, it is not clear how the adaptive mixture approach of \citet{ji2013} scales to very high-dimensional problems.

\clearpage

\renewcommand*{\thefigure}{E.\arabic{figure}}

\section{Additional figures for the illustrative simulated data example of Section~5.1} \label{sec:illustrativeSup}

\begin{figure}[!ht]\centering
\includegraphics[width=0.9\textwidth]{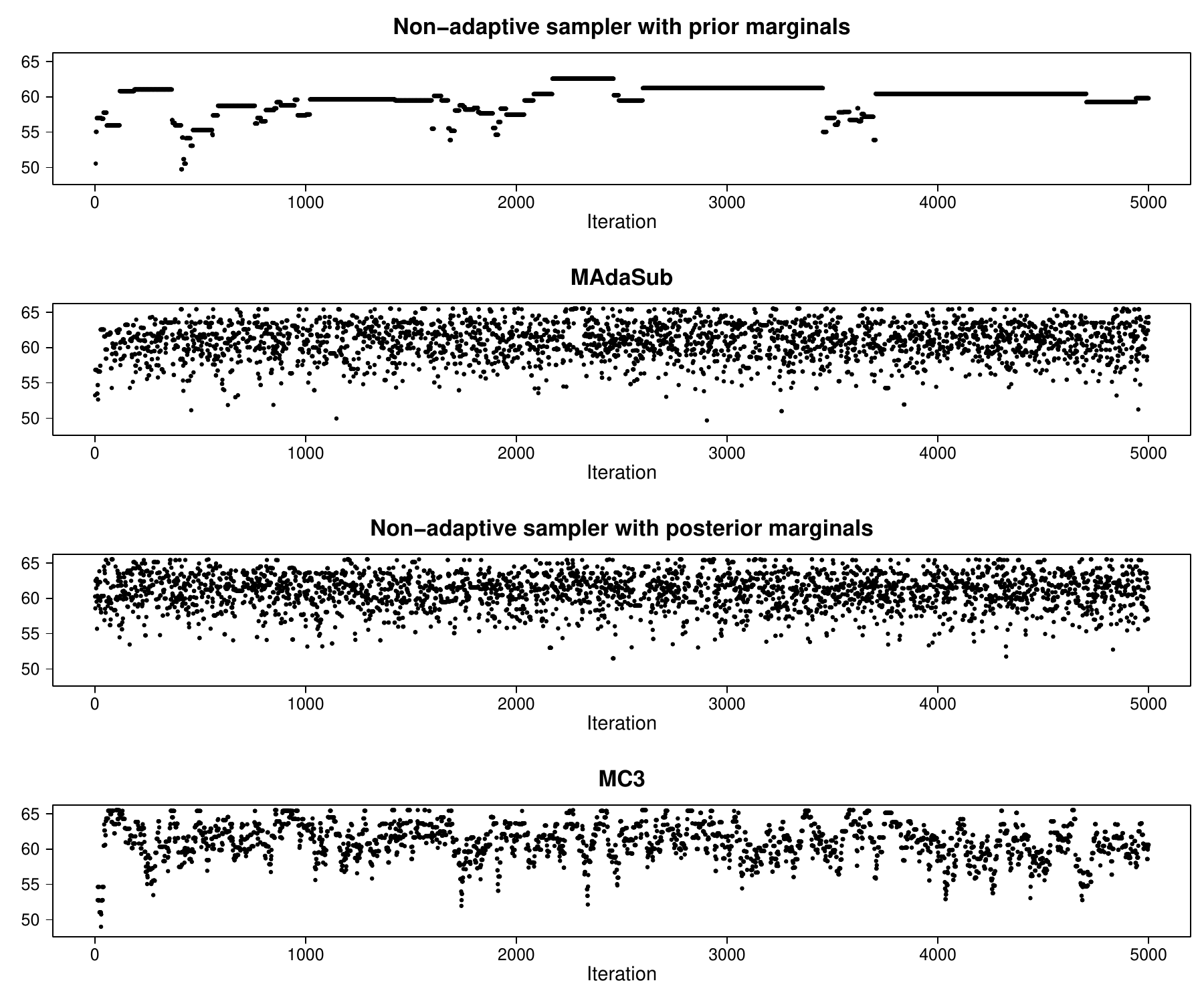} 
\caption{\label{fig:lowdim_values}Illustrative example with g-prior. \small Evolution of the values of the posterior (log-)kernel along the first 5,000 iterations (\(t\)) for non-adaptive sampler with prior marginals as proposal probabilities
, for MAdaSub (with \(L_j=p\)), for non-adaptive sampler with posterior marginals as proposal probabilities and for MC3 sampler (from top to bottom). } 
\end{figure}

\begin{figure}[!t]\centering 
\includegraphics[width=\textwidth]{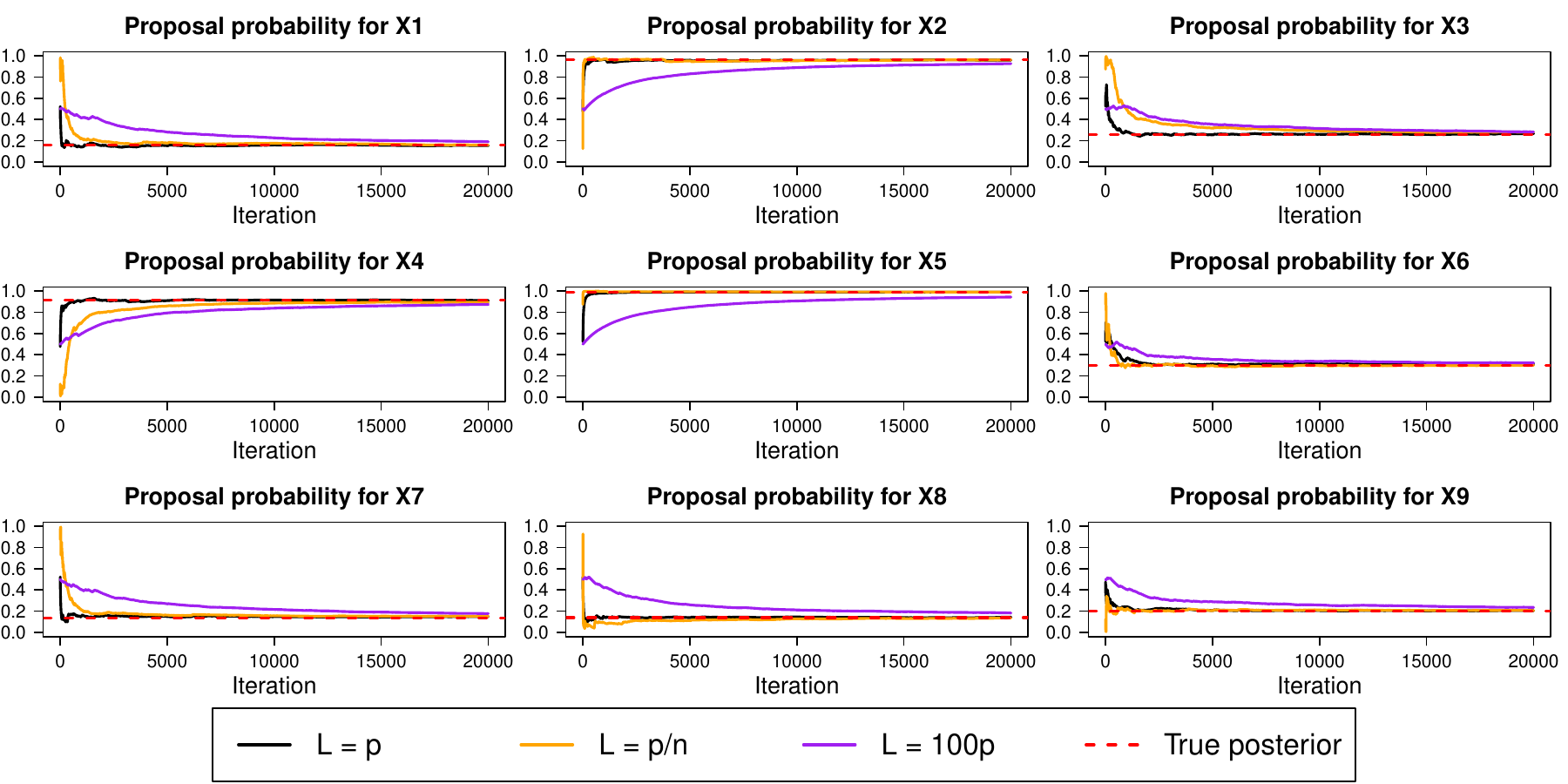}
\caption{\label{fig:lowdim_probabilities} Illustrative example with g-prior. \small Evolution of the proposal probabilities~\(r_j^{(t)}\), for \(j=1,\dots,9\), along the iterations~(\(t\)) of MAdaSub with \(L_j=p\) (black), \(L_j=p/n\) (orange) and \(L_j=100p\) (purple) for \(j\in\mathcal{P}\). The red horizontal lines indicate the true posterior inclusion probabilities.} 
\end{figure}

\begin{figure}[!ht]\centering 
\includegraphics[width=\textwidth]{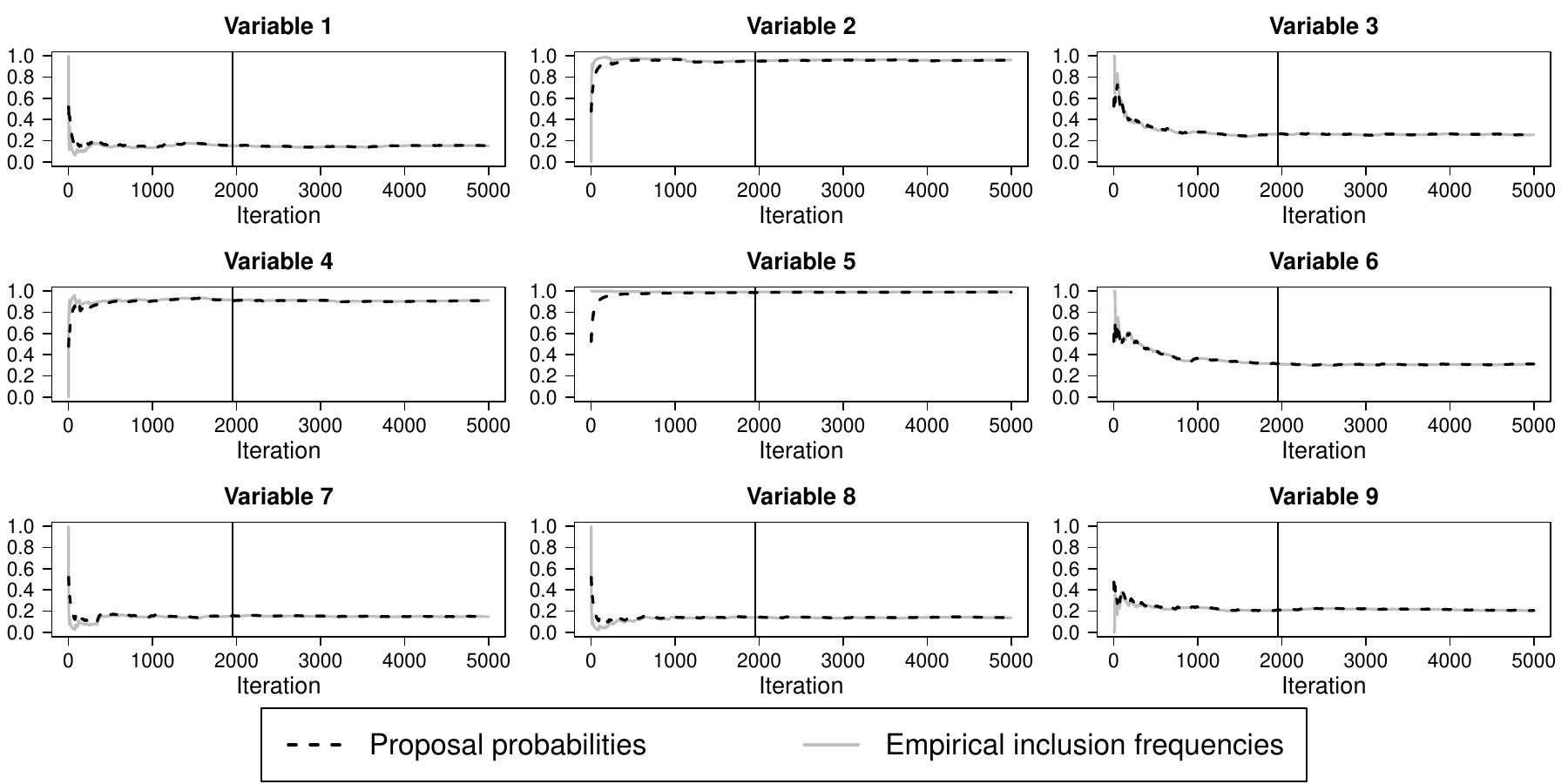}
\caption{\label{fig:lowdim_convergence_check}Illustrative example with g-prior. \small Evolution of proposal probabilities~\(r_j^{(t)}\) and running empirical inclusion frequencies~\(f_j^{(t)}\) along the iterations~(\(t\)) of MAdaSub with \(L_j=p\), for \(j=1,\dots,9\). The vertical line indicates the smallest iteration~\(t_c\) for which \(\max_{j\in\mathcal{P}}|f_j^{(t_c)}-r_j^{(t_c)}|\leq 0.005\).} 
\end{figure}

\clearpage


\setcounter{figure}{0}
\setcounter{table}{0}

\renewcommand*{\thefigure}{F.\arabic{figure}}
\renewcommand*{\thetable}{F.\arabic{table}}

\section{Additional results for the low-dimensional simulation study of Section~5.2}

\begin{figure}[!ht]\centering 
\includegraphics[width=\textwidth]{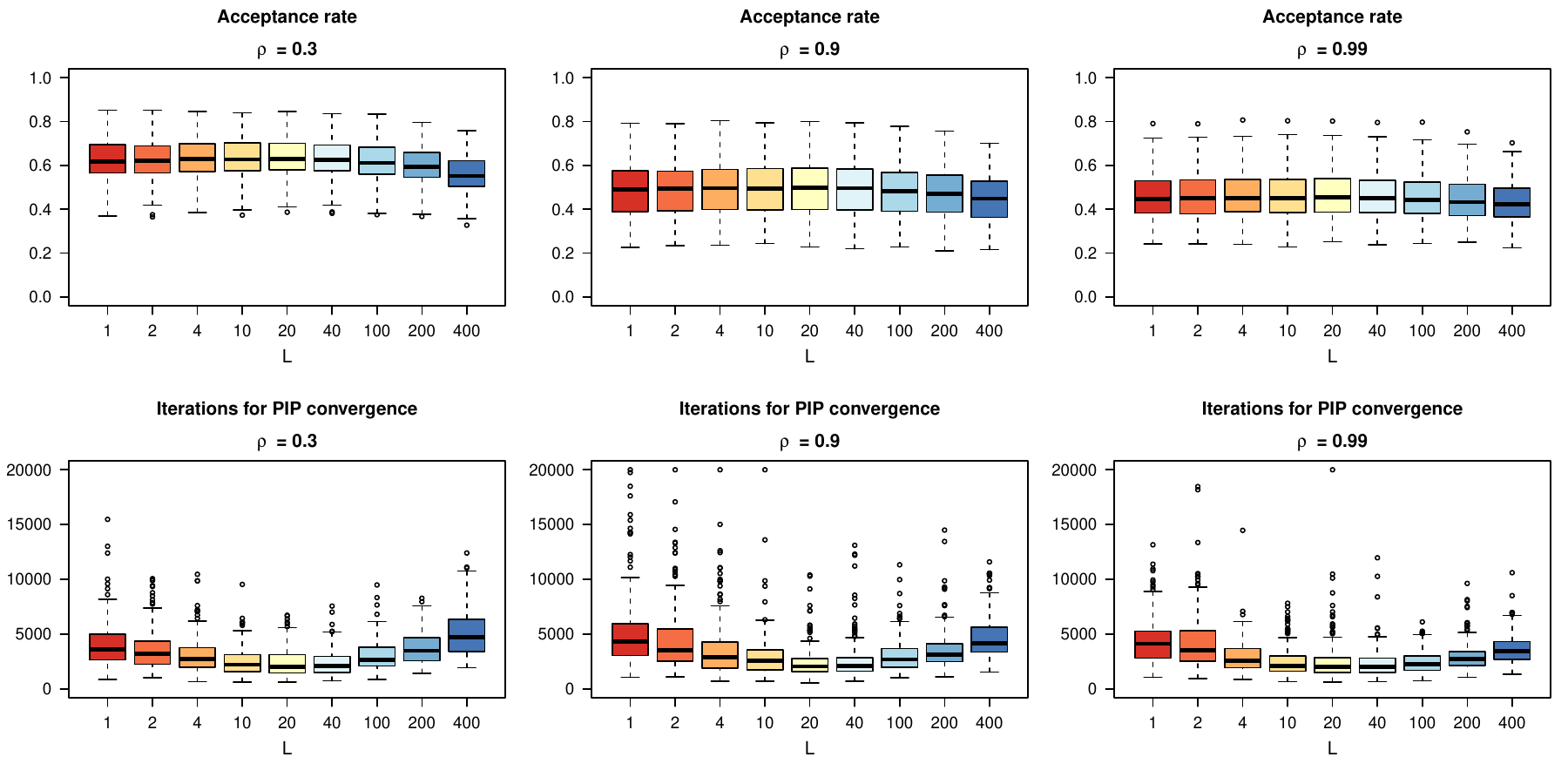}
\caption{\label{fig:simulations_Ljs}\small Results of sensitivity analysis regarding different choices of variance parameters~$L_j=L$ for~$j\in\mathcal{P}$ in MAdaSub for low-dimensional simulation setting with~\(n=60\), \(p=20\) and varying correlation \(\rho\in\{0.3,0.9,0.99\}\) in Toeplitz structure. Initial proposal probabilities \(r_j^{(0)}=0.5\) in MAdaSub are based on prior inclusion probabilities. Performance in terms of acceptance rates (upper plots) and numbers of iterations for convergence of posterior inclusion probabilities~(PIP, lower plots).}  
\end{figure}

\clearpage


\setcounter{figure}{0}
\setcounter{table}{0}

\renewcommand*{\thefigure}{G.\arabic{figure}}
\renewcommand*{\thetable}{G.\arabic{table}}

\section{Additional results for the high-dimensional simulation study of Section~5.3} \label{sec:highdimSup}

In this section we present additional results for the high-dimensional simulation study of Section~5.3 of the main document. \vspace{5mm}

\begin{table*}[!h]
\begin{center}
\resizebox{\textwidth}{!}{\begin{tabular}{ clR{2.8cm}R{2.8cm}R{2.8cm}R{2.8cm} } 
 \hline
 $(n,p)$ & Algorithm & \multicolumn{1}{r}{$\text{SNR}=0.5$}   & \multicolumn{1}{r}{$\text{SNR}=1$} & \multicolumn{1}{r}{$\text{SNR}=2$}  & \multicolumn{1}{r}{$\text{SNR}=3$}  \\ \hline \hline 
 $ (500,500)$        & MAdaSub ser./ par. & 56.6 / \,\,\,20.3      & 6.2 / \,\,\,\,\,\,2.5  & 1.0 / \,\,\,\,2.7 & 2.6 / \,\,\,\,3.8 \\
										 & EIA*/ ASI*         & 4.9 / \,\,\,\,\,\,1.7  & 1.8 / \,\,\,21.3 & 5.5 / \,31.8  & 5.1 / \,\,\,\,7.5 \\ \hline
 $ (500,5000)$       & MAdaSub ser./ par. & 128.6 / 147.8          & 4.8 / \,\,\,\,\,6.0 & 0.8 / \,\,\,\,6.2  & 2.5 / \,\,\,\,9.4 \\
										 & EIA*/ ASI*         & 8.7 / \,\,\,29.9       & 2.2 / 126.9  & 718.0 / 2053  & 81.5 / 2271 \\ \hline
	 $ (1000,500)$     & MAdaSub ser./ par. & 136.7 / \,\,\,71.8     & 4.8 / \,\,\,\,\,\,3.0  & 0.9 / \,\,\,\,3.0  & 2.8 / \,\,\,\,3.8 \\
										 & EIA*/ ASI*         & 5.9 / \,\,\,41.9       & 16.3 / \,\,\,\,\,2.1 & 7.7 / \,16.9  & 4.2 / \,12.0 \\ \hline
	 $ (1000,5000)$    & MAdaSub ser./ par. & 248.2 / 239.5          & 1.0 / \,\,\,\,\,\,0.8 & 2.5 / \,\,\,\,8.5  & 3.6 / \,\,\,\,9.7  \\
										 & EIA*/ ASI*         & 2.2 / \,\,\,15.4       & 2.2 / \,\,\,37.0  & 9167 / 4423  & 11.3 / \,30.8 \\ \hline
\end{tabular}}
\end{center}
\caption{\small Complimentary results of high-dimensional simulation study. Performance of different adaptive algorithms~(A) compared to add-delete-swap \(\text{MC}^3\) algorithm~(B), in terms of the median estimated ratio~$\hat{r}_{A,B}$ of the relative time-standardized effective sample size for PIPs over all variables. Note that, for comparison reasons and in contrast to the median estimated ratios~$\hat{r}_{A,B}^{(20)}$ reported in Table~1 of the paper for the 20 variables with the largest estimated PIPs, here the median is taken over all variables, even though the majority of variables receives very small posterior probability. \\[1mm]
*Results for exploratory individual adaptation~(EIA) and adaptively scaled individual adaptation~(ASI) algorithms are taken from Table~1 in~\cite{griffin2018}. Comparisons between MAdaSub and algorithms of~\cite{griffin2018} should be interpreted in a holistic way, as the used computational systems, implementations and the specific simulated datasets for each setting may differ.}
\label{tab:highdim_supp1}
\end{table*}

Table~\ref{tab:highdim_supp1} provides complimentary results based on the same evaluation metric as in~\cite{griffin2018}, i.e.\ regarding the median estimated ratio of the relative time-standardized effective sample size for PIPs over \textbf{all} variables. Results indicate that MAdaSub also yields a competitive performance compared to the exploratory individual adaptation~(EIA) and adaptively scaled individual adaptation~(ASI) algorithms of~\cite{griffin2018}, with advantages of MAdaSub in low SNR settings and advantages of the adaptive algorithms of~\cite{griffin2018} in high SNR settings.

\begin{table*}
\begin{center}
\resizebox{\textwidth}{!}{\begin{tabular}{ llR{2.8cm}R{2.8cm}R{2.8cm}R{2.8cm} } 
 \hline
  & & \multicolumn{1}{c}{$\text{SNR}=0.5$} & \multicolumn{1}{c}{$\text{SNR}=1$} & \multicolumn{1}{c}{$\text{SNR}=2$} & \multicolumn{1}{c}{$\text{SNR}=3$}  \\ 
 $L_j$ & MAdaSub & $\hat{r}_{A,B}^{(20)}$ / \,\,\,Acc.  &  $\hat{r}_{A,B}^{(20)}$ / \,\,\,Acc. &  $\hat{r}_{A,B}^{(20)}$ / \,\,\,Acc. & $\hat{r}_{A,B}^{(20)}$ / \,\,\,Acc.  \\ \hline \hline 
 $ 1$        & serial &   63.3 / 44.9\%   & 21.9 / 33.4\%      &  5.3 / \,\,\,9.6\%  &  8.0  / 13.2\% \\
             & parallel & 19.7 / 45.2\%   & 8.9 /  37.4\%      &  8.6 / 20.4\%       &  13.0 / 24.7\% \\ 			\hline
 $ 5$        & serial & 72.0 / 45.3\%     & 21.3 / 35.7\%      &  7.0 / 12.2\%       &  12.9 / 16.8\% \\
             & parallel & 20.9 / 45.3\%   & 8.3 / 38.3\%       &  12.1 / 23.9\%      &  17.0 / 29.2\% \\ 			\hline
 $ 50$       & serial & 66.4 / 45.2\%     & 23.8 / 36.6\%      &  9.4 / 14.0\%       &  16.8 / 19.7\% \\
             & parallel & 20.5 / 45.3\%   & 7.8 / 38.9\%       &  14.8 / 26.6\%      &  20.2 / 31.9\% \\ 			\hline
 $ 500$      & serial & 67.0 / 44.6\%     & 22.7 / 31.9\%      &  4.7 / \,\,\,6.3\%  &  8.0  / \,\,\,9.3\% \\
             & parallel & 22.2 / 45.2\%   & 8.5 / 36.9\%       &  9.0 / 18.3\%       &  10.4 / 20.9\% \\ 			\hline				
 $ 5000$     & serial & 27.6 / 26.4\%     & 3.9 / \,\,\,7.8\%  &  0.2 / \,\,\,0.2\%  &  0.1 / \,\,\,0.2\% \\
             & parallel & 21.3 / 44.0\%   & 8.8 / 28.7\%       &  1.2 / \,\,\,4.0\%  &  1.8 / \,\,\,4.8\% \\ 			\hline		
 $ 50000$    & serial & 0.3 / \,\,\,1.0\% & 0.2 / \,\,\,0.2\%  &  0.03 / \,\,\,0.1\% &  0.02 / \,\,\,0.1\% \\
             & parallel & 3.8 / 13.1\%    & 0.3 / \,\,\,2.3\%  &  0.02 / \,\,\,0.1\% &  0.01 / \,\,\,0.1\% \\ 			\hline		
\end{tabular}}
\end{center}
\caption{\small Results of sensitivity analysis regarding different choices of variance parameters~$L_j$ in MAdaSub for high-dimensional simulation setting with \(n=500\) and \(p=500\), with fixed choices of~$r_j^{(k,0)}=10/p$ for all serial and parallel chains~$k$. Performance of MAdaSub algorithms~(A) with serial and parallel updating schemes compared to add-delete-swap \(\text{MC}^3\) algorithm~(B) in terms of median estimated ratios~$\hat{r}_{A,B}^{(20)}$ of the relative time-standardized effective sample size for PIPs over the 20 variables with the largest estimated PIPs, and in terms of median acceptance rates (Acc.). 
}
\label{tab:highdim_Lj}
\end{table*}

Table~\ref{tab:highdim_Lj} provides results of a sensitivity analysis regarding different choices of the variance parameters~$L_j$ in MAdaSub for the high-dimensional simulation setting of Section~5.3 with \(n=500\) and \(p=500\), showing that the choice~$L_j=p=500$ also performs well for all considered signal-to-noise ratios \(\text{SNR}\in\{0.5,1,2,3\}\); however, competitive and partly favourable results are also obtained for \(L_j<p=500\) in this sparse high-dimensional setting.

\begin{table*}
\begin{center}
\resizebox{\textwidth}{!}{\begin{tabular}{ lR{2.8cm}R{2.8cm}R{2.8cm}R{2.8cm} } 
 \hline
  &  \multicolumn{1}{c}{$\text{SNR}=0.5$} & \multicolumn{1}{c}{$\text{SNR}=1$} & \multicolumn{1}{c}{$\text{SNR}=2$} & \multicolumn{1}{c}{$\text{SNR}=3$}  \\ 
 Initialization & $\hat{r}_{A,B}^{(20)}$ / \,\,\,Acc.  &  $\hat{r}_{A,B}^{(20)}$ / \,\,\,Acc. &  $\hat{r}_{A,B}^{(20)}$ / \,\,\,Acc.  & $\hat{r}_{A,B}^{(20)}$ / \,\,\,Acc. \\ \hline \hline 
 Fixed $r_j^{(k,0)}$ \& fixed $L_j^{(k)}$        & 20.8  / 45.2\%  & 10.3 / 36.9\%  & 10.4 / 18.6\% &  12.8 / 21.4\%  \\
 Random $r_j^{(k,0)}$ \& fixed $L_j^{(k)}$  & 21.0  / 45.3\%  & 9.9 / 37.8\%   & 10.5 / 19.5\% &  15.5 / 22.8\% \\
 Random $r_j^{(k,0)}$ \& random $L_j^{(k)}$       & 20.0  / 45.3\%  & 9.0 / 37.7\%   & 7.8  / 18.1\% &  10.9 / 21.4\% \\
\end{tabular}}
\end{center}
\caption{\small Results of sensitivity analysis regarding random (different) versus fixed (the same) initialisations of tuning parameters~$r_j^{(k,0)}$ and~\(L_j^{(k)}\) for the parallel MAdaSub chains in the high-dimensional simulation setting with \(n=500\) and \(p=500\). Fixed initializations are $r_j^{(k,0)}=10/p$ and $L_j^{(k)}=p$, while  random initializations are \(r_j^{(k,0)}=q^{(k)}/p\sim U(2/p,10/p)\) and \(L_j^{(k)}=L^{(k)}\sim U(p/2,2p)\) for each chain~$k$. Performance of parallel MAdaSub algorithm~(A) compared to add-delete-swap \(\text{MC}^3\) algorithm~(B) in terms of median estimated ratios~$\hat{r}_{A,B}^{(20)}$ of the relative time-standardized effective sample size for PIPs over the 20 variables with the largest estimated PIPs and in terms of median acceptance rates (Acc.). 
}
\label{tab:highdim_randomfixed}
\end{table*}

Table~\ref{tab:highdim_randomfixed} provides results of a sensitivity analysis regarding random (different) versus fixed (the same) initialisations of the tuning parameters~$r_j^{(k,0)}$ and~\(L_j^{(k)}\) for the parallel MAdaSub chains in the same high-dimensional simulation setting, showing that the performance of MAdaSub appears not to be largely affected by the different (random or fixed) initializations of its tuning parameters in this setting. Yet, results indicate that choosing different random initial proposal probabilities~\(r_j^{(k,0)}\) for the chains~$k$ can be beneficial and tends to yield slightly improved performance compared to considering the same fixed tuning parameters~\(r_j^{(k,0)}\) and~$L_j^{(k)}$ for each chain. On the other hand, the parallel MAdaSub algorithm with random initializations of both tuning parameters~\(r_j^{(k,0)}\) and~$L_j^{(k)}~\sim U(p/2,2p)$ tends to yield slightly worse performance, as variance parameters~$L_j^{(k)}>p$ are not favourable in this setting (see also Table~G.2). Despite this, to avoid optimistic biases in the evaluation of the proposed algorithm (cf.~\citealp{buchka2021}), in Table~1 of the main document we still report the results for the parallel version with the originally considered random initializations of both tuning parameters~\(r_j^{(k,0)}\) and~$L_j^{(k)}$. 

\begin{table*}
\begin{center}
\resizebox{\textwidth}{!}{\begin{tabular}{ cR{4cm}R{4cm}R{4cm}R{4cm} } 
 \hline
  &  \multicolumn{1}{c}{$\text{SNR}=0.5$} & \multicolumn{1}{c}{$\text{SNR}=1$} & \multicolumn{1}{c}{$\text{SNR}=2$} & \multicolumn{1}{c}{$\text{SNR}=3$}  \\ 
 $(R,T)$ & $\hat{r}_{A,B}^{(20)}$ / \,\,\,Acc. / \,\,\,Time  &  $\hat{r}_{A,B}^{(20)}$ / \,\,\,Acc. / \,\,\,Time &  $\hat{r}_{A,B}^{(20)}$ / \,\,\,Acc. / \,\,\,Time & $\hat{r}_{A,B}^{(20)}$ / \,\,\,Acc. / \,\,\,Time \\ \hline \hline 
 $ (10,5000)$       &  42.5 / 45.4\% / \,\,\,32.4s  & 17.0 / 34.9\% / \,\,\,39.3s  & 6.3 / 10.2\% / \,\,\,53.8s  &  10.1 / 14.6\% / \,\,\,55.6s \\
 $ (20,2500)$       &  36.3 / 45.3\% / \,\,\,36.6s  & 15.3 / 36.1\% / \,\,\,43.9s  & 7.5 / 13.6\% / \,\,\,58.8s  &  11.5 / 18.0\% / \,\,\,60.1s \\
 $ (50,1000)$       &  23.0 / 45.3\% / \,\,\,54.4s  & 9.9  / 37.7\% / \,\,\,60.8s  & 7.8 / 18.1\% / \,\,\,77.6s  &  12.2 / 21.4\% / \,\,\,79.1s \\
 $ (100,500)$       &  16.3 / 45.3\% / \,\,\,81.9s  & 5.4  / 38.8\% / \,\,\,90.1s  & 8.0 / 22.2\% / 108.7s &  11.3 / 25.3\% / 109.7s \\
 $ (200,250)$       &  10.0 / 45.3\% / 136.4s       & 3.6  / 39.7\% / 146.5s       & 6.1 / 24.3\% / 171.5s &  8.6  / 28.7\% / 174.6s \\
\end{tabular}}
\end{center}
\caption{\small Results of sensitivity analysis regarding different choices of rounds~$R$ and iterations~$T$ per round in parallel version of MAdaSub for high-dimensional simulation setting with \(n=500\) and \(p=500\). Performance of parallel MAdaSub algorithm~(A) compared to add-delete-swap \(\text{MC}^3\) algorithm~(B) in terms of median estimated ratios~$\hat{r}_{A,B}^{(20)}$ of the relative time-standardized effective sample size for PIPs over the 20 variables with the largest estimated PIPs, in terms of median acceptance rates (Acc.) and in terms of median computation times (in seconds). 
}
\label{tab:highdim_RT}
\end{table*}

Finally, Table~\ref{tab:highdim_RT} provides results of a sensitivity analysis regarding different choices of the number of rounds~$R$ and the number of iterations~$T$ per round in the parallel version of MAdaSub for the same high-dimensional simulation setting, considering varying combinations of \((R,T)\) such that the total number of iterations~$R\times T$ per chain remains constant. Results show that there is a trade-off regarding sampling effectiveness and computational efficiency: if the frequency of communication between the different chains is increased (i.e.\ larger numbers of rounds~$R$), then the convergence of the proposal probabilities is accelerated, leading to larger acceptance rates (for \(\text{SNR}\geq1\)); however, the higher frequency of communication between the chains comes at the prize of larger computation times. For settings with high signal-to-noise ratios (\(\text{SNR}\geq2\)), the resulting median estimated ratios of the relative time-standardized effective sample size are largest for~$R\in[20,100]$. Note that we considered the number of parallel chains to be the same as the number of assigned CPUs (i.e.\ 5 parallel chains with 5 CPUs, see Section~5.3), which is the most natural choice. However, in practice the ``optimal'' choice of the number of rounds~$R$ may also depend on the number of available CPUs for parallel computation (especially in case this number is considerably different from the number of parallel MAdaSub chains).

\clearpage


\setcounter{figure}{0}

\renewcommand*{\thefigure}{H.\arabic{figure}}

\section{Additional results for Tecator data application of Section~6.1} \label{sec:TecatorSup}

Here we provide additional results regarding the efficiency of the serial MAdaSub algorithm under the same setting as in \citet{lamnisos2013}, where several adaptive and non-adaptive MCMC algorithms are compared using normal linear models for the Tecator data. In particular, \citet{lamnisos2013} consider a classical \(\text{MC}^3\) algorithm (\citealp{madigan1995}), the adaptive Gibbs sampler of \citet{nott2005} and adaptive and non-adaptive Metropolis-Hastings algorithms based on the tunable model proposal of \citet{lamnisos2009}. In the comparative study of \citet{lamnisos2013} each algorithm is run for 2,000,000 iterations, including an initial burn-in period of 100,000 iterations. Furthermore, thinning is applied using only every 10th iteration, so that the finally obtained MCMC sample has size 190,000. For comparison reasons, after a burn-in period of 100,000 iterations, we run the serial MAdaSub algorithm for 190,000 iterations, so that the considered MCMC sample has the same size as in \citet{lamnisos2013}. In the serial MAdaSub algorithm we set \(r_j^{(0)}=\frac{5}{100}\) for \(j\in\mathcal{P}\), i.e.\ we use the prior inclusion probabilities as the initial proposal probabilities in MAdaSub; further, we set \(L_j=p\) for \(j\in\mathcal{P}\) and \(\epsilon=\frac{1}{p}\). Since the acceptance rate of MAdaSub is already sufficiently large in the considered setting yielding a well-mixing algorithm, we do not consider additional thinning of the resulting chain. In fact, the acceptance rate of the serial MAdaSub chain is approximately 0.38 for the 190,000 iterations (excluding the burn-in period). We note that in this example the relatively large number of 100,000 burn-in iterations is not necessarily required for MAdaSub and is only used for comparison reasons.   

\citet{lamnisos2013} report estimated median effective sample sizes of the different samplers 
for the evolution of the indicators \(\big(\gamma_j^{(t)}\big)_{t=1}^T\) for \(j\in\mathcal{P}\), where \(\gamma_j^{(t)}=\mathbbm{1}_{S^{(t)}}(j)\) indicates whether variable~\(X_j\) is included in the sampled model~\(S^{(t)}\) in iteration~\(t\). The estimated median effective sample size for the 190,000 iterations of the serial MAdaSub algorithm is approximately 38,012 (using the R-package \texttt{coda}), which is slightly larger than the values for the competing algorithms reported in \citealp{lamnisos2013} (the largest one is 37,581 for the ``optimally'' tuned Metropolis-Hastings algorithm). Note that when using 1,900,000 iterations with thinning (every 10th iteration after 100,000 burn-in iterations) as in the other algorithms, the estimated median effective sample size for MAdaSub is much larger (178,334), yielding almost independent samples of size 190,000.  

We finally provide details on the computational costs of the serial and parallel versions of MAdaSub for the analysis of the Tecator data presented in Section~6.1 of the main document. The computation time for each of the 5000 iterations of the serial MAdaSub algorithm is approximately 3.5 seconds (using an R implementation of MAdaSub on an Intel(R) Core(TM) i7-7700K, 4.2 GHz processor); thus, even without parallelization, one obtains accurate posterior estimates with the serial MAdaSub algorithm within seconds using a usual desktop computer (e.g.\ after 10,000 or 15,000 iterations, see Figure~4 of the main document). \citet{lamnisos2013} report that the computation times for each of the other considered MCMC methods were in the order of 25,000 seconds for the total number of 2,000,000 iterations (using a MATLAB implementation).  Although the computation times are not directly comparable, these results indicate that the serial MAdaSub algorithm is already very efficient. The timings for MAdaSub are also of a similar order as for the recent adaptive algorithms of~\cite{griffin2018}, who report that short runs of 6000 iterations of the exploratory individual adaptation algorithm yield stable estimates for the Tecator data with computation times of about 5 seconds~\citep{griffin2018}. When using a computer cluster with 50 CPUs, the overall computation time for all considered 50 MAdaSub chains (each with a large number of 290,000 iterations) is 460 seconds, while the computation time for a single chain is 231 seconds on the same system. This shows that, even though 25 of the 50 MAdaSub chains communicate with each other after every 5,000 iterations, the parallelization yields a substantial speed-up in comparison to a serial application of 50 independent chains.

\clearpage


\setcounter{figure}{0}

\renewcommand*{\thefigure}{I.\arabic{figure}}

\section{Additional results for PCR and Leukemia data applications of Section~6.2} \label{sec:PCRSup}

To further illustrate the stability of the results, we examine three independent runs of the serial MAdaSub algorithm for the PCR and leukemia data, each with \(T=1{,}000{,}000\) iterations, setting \(r_j^{(0)}=\frac{q}{p}\) as initial proposal probabilities with different expected search sizes \(q\): for the first run we set \(q=2\), for the second run \(q=5\) and for the third run \(q=10\). Further tuning parameters are set to \(L_j=p\) and \(\epsilon=1/p\) for each of the three MAdaSub runs.

\begin{figure}[!ht]
\centering
\begin{subfigure}{.5\textwidth}
  \centering
		\includegraphics[width=\linewidth]{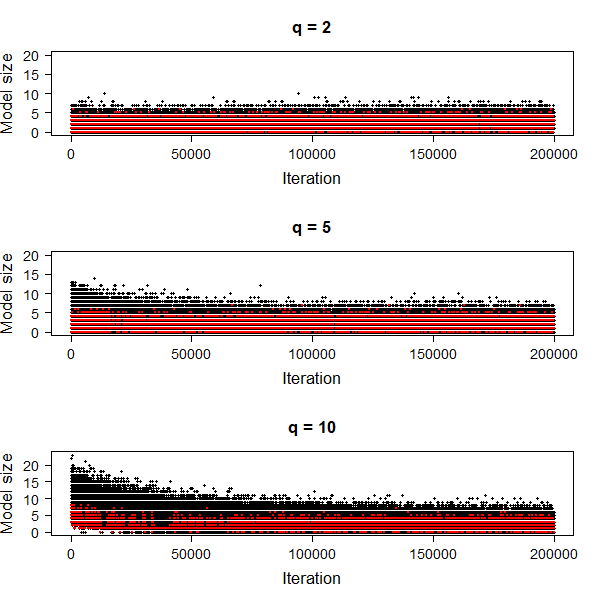}
  \caption{PCR dataset}
\end{subfigure}%
\begin{subfigure}{.5\textwidth}
  \centering
	\includegraphics[width=\linewidth]{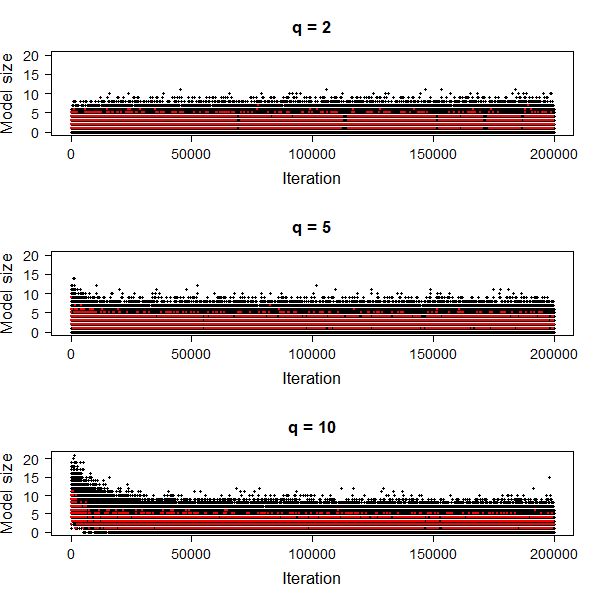} 
  \caption{Leukemia dataset}
\end{subfigure}
\caption{PCR and leukemia data applications. \small Evolution of the sizes~\(|V^{(t)}|\) of the proposed models (black) and of the sizes~\(|S^{(t)}|\) of the sampled models (red) along the first 200,000 iterations~(\(t\)) of three independent runs of the serial MAdaSub algorithm for \(q=2\), \(q=5\) and \(q=10\).} 
\label{fig:highdimreal_sizes}
\end{figure}

Figure~\ref{fig:highdimreal_sizes} depicts the evolution of the sizes of the sampled and proposed models for the first 200,000 iterations of MAdaSub, showing that the algorithm quickly adjusts the search sizes appropriately based on the history of the already sampled models. Furthermore, Figure~\ref{fig:highdimreal_scatter} shows scatterplots of the final proposal probabilities for the different runs of MAdaSub, illustrating that the proposal probabilities converge to the same values despite their different initial choices (with somewhat larger variability for the leukemia data). 

\begin{figure}[!ht] 
\centering
\begin{subfigure}[t]{1\textwidth}
 \includegraphics[width=\linewidth]{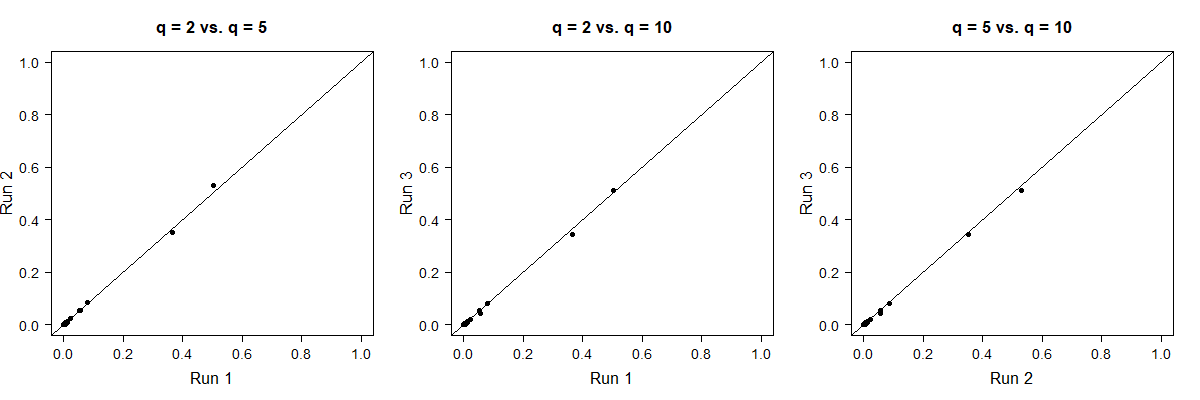}
   \caption{PCR dataset}
	\vspace{0.5cm}
\end{subfigure}
   \begin{subfigure}[t]{1\textwidth}
		\includegraphics[width=\linewidth]{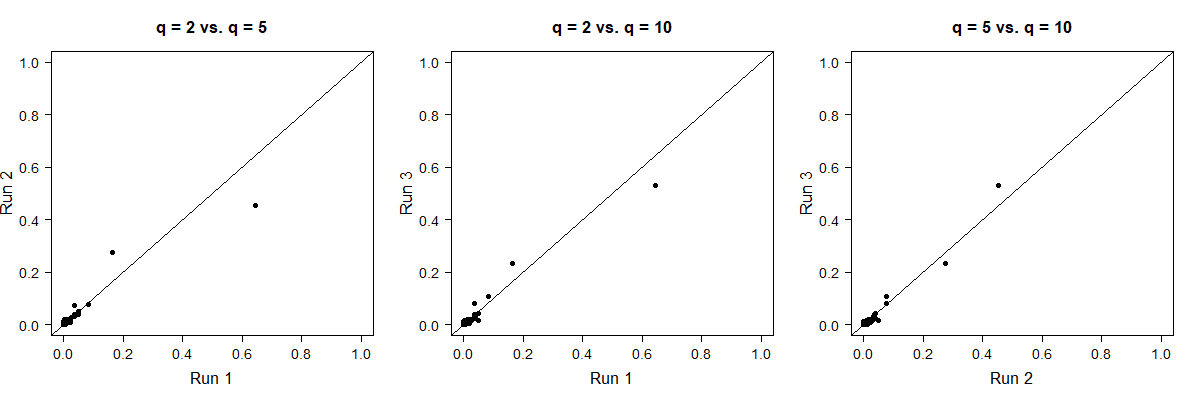}
   \caption{Leukemia dataset}
\end{subfigure}
\caption{PCR and leukemia data applications. \small Scatterplots of final proposal probabilities~\(r_j^{(T)}\) after \(T=1{,}000{,}000\) iterations for three independent runs of the serial MAdaSub algorithm (\(q=2\), \(q=5\) and \(q=10\)).} 
\label{fig:highdimreal_scatter}
\end{figure}

Similarly as for the Tecator data, Figures~\ref{fig:PCR_parallel} and~\ref{fig:Leukemia_parallel} depict boxplots of empirical inclusion frequencies of the most informative variables for the first three rounds (each of 20,000 iterations) and after 1,000,00 iterations (with a burn-in period of 200,000 iterations) of 25 serial and 25 parallel MAdaSub chains for the PCR and leukemia dataset, respectively, considering random initializations of proposal probabilities \(r_j^{(k,0)}\) and variance parameters \(L_j^{(k)}\) for each chain \(k=1,\dots,50\) (see the main paper for details). The results further illustrate the benefits of the parallel version of MAdaSub, providing particularly stable estimates of posterior inclusion probabilities for the PCR data after only 60,000 iterations (see also Figures~5 and~6 of the main document).

\begin{figure}[!t]\centering 
\includegraphics[width=\textwidth]{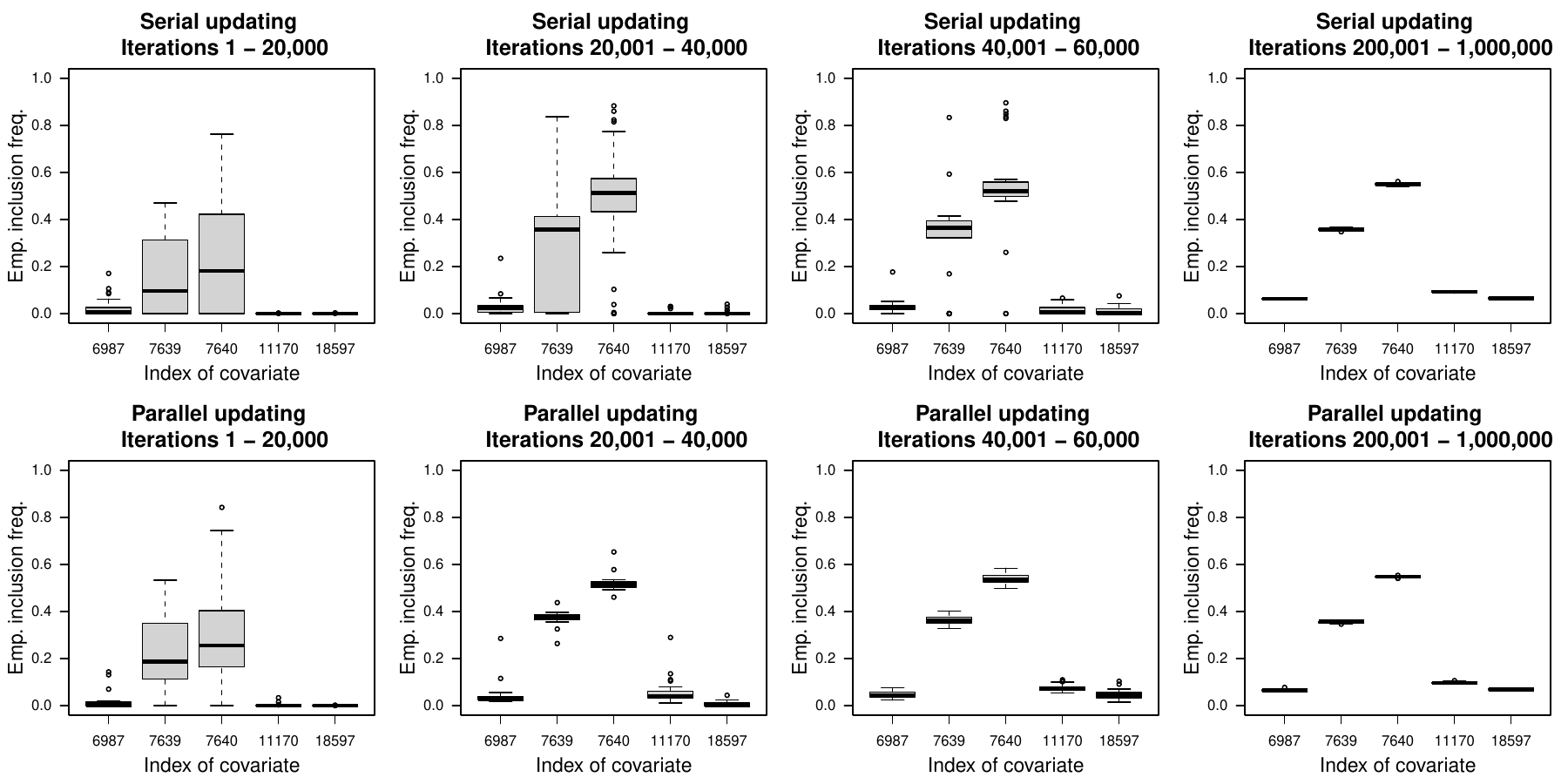}
\caption{\label{fig:PCR_parallel}PCR data application. \small Results of 25 serial MAdaSub chains (Algorithm~1, top) and of 25 parallel MAdaSub chains exchanging information every 20,000 iterations (Algorithm~\ref{algo:parallel}, bottom) in terms of empirical variable inclusion frequencies \(f_j\) for most informative variables \(X_j\) (with final \(f_j\geq 0.05\) for at least one chain).} 
\end{figure}

\begin{figure}[!t]\centering 
\includegraphics[width=\textwidth]{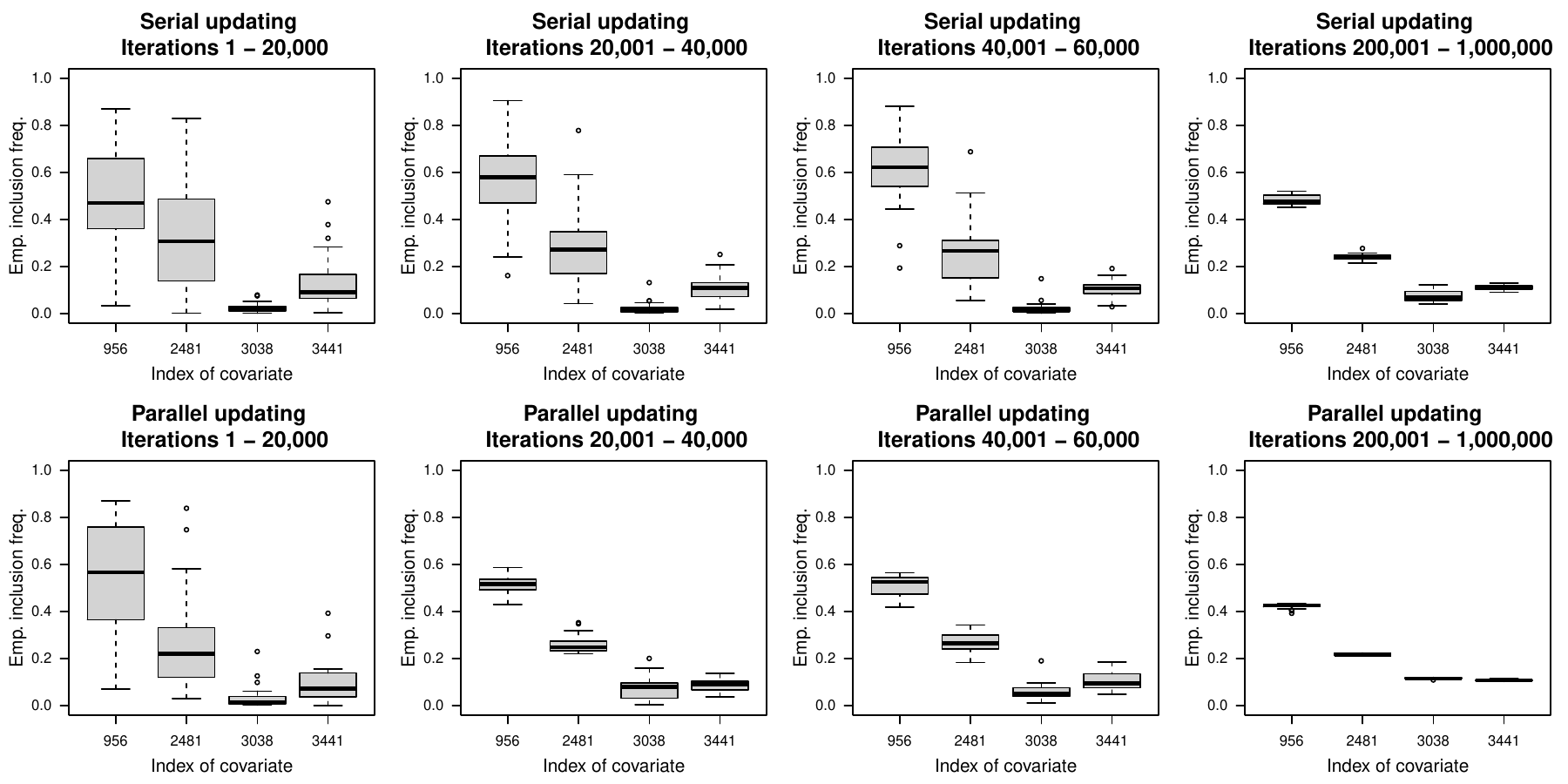}
\caption{\label{fig:Leukemia_parallel}Leukemia data application. \small Results of 25 serial MAdaSub chains (Algorithm~1, top) and of 25 parallel MAdaSub chains exchanging information every 20,000 iterations (Algorithm~\ref{algo:parallel}, bottom) in terms of empirical variable inclusion frequencies \(f_j\) for most informative variables \(X_j\) (with final \(f_j\geq 0.1\) for at least one chain).} 
\end{figure}

\end{document}